\documentclass[12pt,reqno]{amsart}

\usepackage{epsfig}
\usepackage{amsmath, amsthm, amsfonts}
\usepackage{graphicx}

\numberwithin{equation}{section}

\newcommand{\di}{\displaystyle}

\newcommand{\C}{{\mathbb C}}

\newcommand{\Z}{{\mathbb Z}}

\newcommand{\al}{\alpha}
\newcommand{\be}{\beta}
\newcommand{\ga}{\gamma}

\newcommand{\la}{\lambda}
\newcommand{\ep}{\varepsilon}
\newcommand{\de}{\delta}
\newcommand{\De}{\Delta}

\newcommand{\sg}{\sigma}

\newcommand{\bigO}{{\mathcal O}}
\newtheorem{theo}{{\sc \bf Theorem}}[section]

\newtheorem{lem}[theo]{{\sc \bf Lemma}}
\newtheorem{prop}[theo]{{\sc \bf Proposition}}

\newenvironment{rem}{\medskip\noindent{\it Remark:\/} }{\medskip}

\begin{document}

\title[Six-vertex model with partial DWBC]
{Six-vertex model with partial domain wall boundary conditions: ferroelectric phase}

\author{Pavel Bleher}
\address{Department of Mathematical Sciences,
Indiana University-Purdue University Indianapolis,
402 N. Blackford St., Indianapolis, IN 46202, U.S.A.}
\email{bleher@math.iupui.edu}

\author{Karl Liechty}
\address{Department of Mathematical Sciences,
DePaul University,
Chicago, IL 60614, U.S.A.}
\email{kliechty@depaul.edu}

\thanks{The first author is supported in part
by the National Science Foundation (NSF) Grant DMS-1265172.}

\date{\today}

\begin{abstract}
We obtain an asymptotic formula for the partition function of the six-vertex model with partial
domain wall boundary conditions in the ferroelectric phase region. The proof is based on a formula for the partition function involving the determinant of a matrix of mixed Vandermonde/Hankel type. This determinant can be expressed in terms of a system of discrete orthogonal polynomials, which can then be evaluated asymptotically by comparison with the Meixner polynomials.
\end{abstract}

\maketitle

\section{Introduction}

We consider the the six-vertex model on a rectangular lattice of size $(n-m) \times n$ for any positive integer $n$ and any integer $m$ with $0\le m <n$. The states of the model are realized by placing arrows on edges of the lattice obeying the {\it ice rule}, meaning that at each vertex there are exactly two arrows pointing in and two arrows pointing out. There are six possible configurations of arrows at each vertex, and we label the six vertex types as shown in Fig. \ref{arrows}. The {\it partial domain wall boundary conditions} (pDWBC) are defined in the following way. On the left and right boundaries all arrows point out of the lattice, and on the bottom boundary all arrows point in.  The top boundary is free, and the ice-rule implies that there are exactly $m$ arrows pointing out on this boundary, and the remaining $(n-m)$ arrows point in. In Fig. \ref{arrow_conf} below, an example of the arrow configuration with the partial domain wall boundary conditions is shown on the $3\times 5$ lattice.  

For each of the six vertex types we assign a weight $w_i$, $i=1, \dots, 6$, and define the weight of an arrow configuration as the product of the weights of the vertices in the configuration. That is, for a configuration $\sg$ of arrows, its weight $w(\sg)$ is defined to be
\begin{equation}\label{int1}
w(\sg)=\prod_{x\in V_{n-m,n}} w_{t(x; \sg)}=\prod_{j=1}^6 w_i^{N_i(\sg)}\,,
\end{equation}
where $V_{n-m,n}$ is the set of vertices in the lattice, $t(x; \sg)$ is the type of vertex at the vertex $x\in V_{n-m,n}$ in the configuration $\sg$, and $N_i(\sg)$ is the number of vertices of type $i$ in the configuration $\sg$. The Gibbs measure on states is then defined as 
\begin{equation}\label{int2}
\mu(\sg)= \frac{w(\sg)}{Z_{n-m,n}}\,, \qquad Z_{n-m,n}\equiv Z_{n-m,n}(w_1, w_2, w_3, w_4, w_5, w_6) =\sum_{\sg} w(\sg),
\end{equation}
where $Z_{n-m,n}$ is the {\it partition function}, and the sum is over all configurations obeying the pDWBC.

\begin{figure}
\begin{center}\scalebox{0.55}{\includegraphics{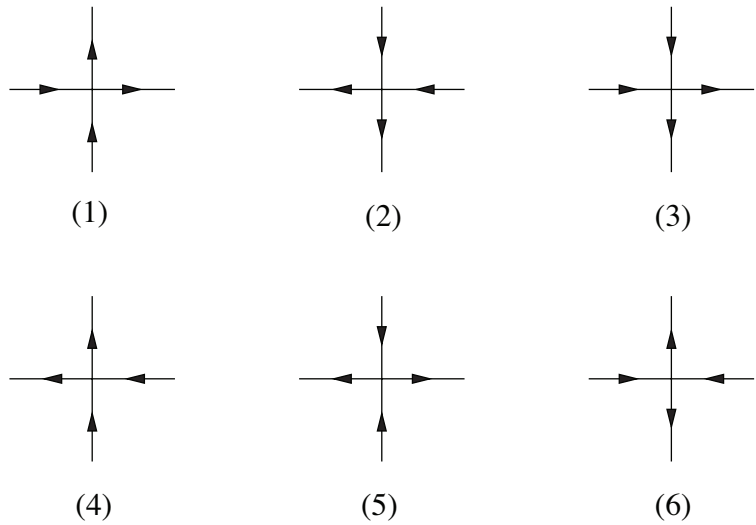}}\end{center}
\caption{The six types of vertices allowed under the ice-rule.}
\label{arrows}
\end{figure}

\begin{figure}
\begin{center}\scalebox{0.45}{\includegraphics{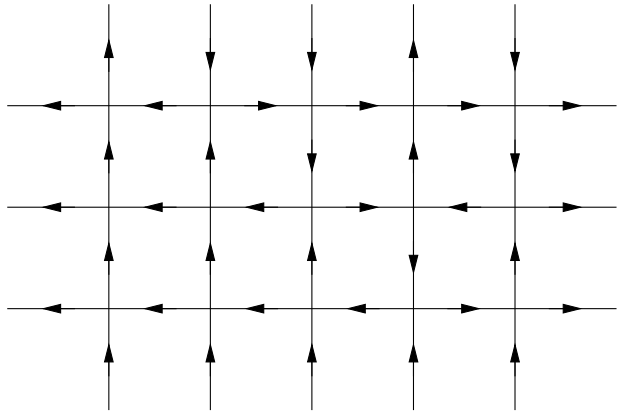}}\end{center}
\caption{An example of the arrow configuration with the partial domain wall boundary conditions on the $3\times 5$ lattice.}
\label{arrow_conf}
\end{figure}
When $m=0$, the pDWBC reduces to the {\it domain wall boundary conditions} on the $n\times n$ lattice \cite{Korepin82}, and the asymptotic expansion of the partition function $Z_{n,n}$ as $n\to \infty$ has been studied in detail in a series of papers by the first author of the current paper and various coauthors. For a complete description, see the monograph \cite{Bleher-Liechty14} of the current authors. The main purpose of the current work is to obtain an asymptotic expansion for the pDWBC partition function as well. Let us also note that the pDWBC partition function has recently appeared in the literature as an expression for certain quantities related to the XXX spin chain, and related models in mathematical physics, see \cite{Kostov-Matsuo12}, \cite{Foda-Wheeler12}, and references therein.

\subsection{Conservation laws and reduction of parameters}
A priori the six-vertex model has six parameters: the weights $w_1, \dots, w_6$. By observing some quantities which are conserved in each state, we can reduce the number of parameters to three. Namely we have the following proposition.
\begin{prop}\label{conlaws} In the six-vertex model on the $(n-m)\times n$ lattice, the following equations hold for every state $\sg$ satisfying pDWBC:
\begin{equation}\label{int3}
\begin{aligned}
N_1(\sg)+N_2(\sg)+N_3(\sg)+N_4(\sg)+N_5(\sg)+N_6(\sg)&=n(n-m), \\
N_5(\sg)-N_6(\sg)&=n-m, \\
N_1(\sg)-N_2(\sg)+N_4(\sg)-N_3(\sg)&=m(n-m). \\
\end{aligned}
\end{equation}
\end{prop}
The first equation in \eqref{int3} is trivial and simply counts the total number of vertices. The second equation follows from the fact that in each row there is one more type-5 vertex than type-6 vertex, which is a direct consequence of the domain wall boundary condition in each row. We prove the third equation in Appendix \ref{app0} in the end of the paper.

\begin{rem}
In the case $m=0$, the third equation of \eqref{int3} can be split into the two equations $N_1(\sg)-N_2(\sg)=N_4(\sg)-N_3(\sg)=0$, giving four conserved quantities. In this case the general six-vertex model can be reduced to two parameters, see \cite{Bleher-Liechty09}, \cite{Bleher-Liechty14}.
\end{rem}

Let us now discuss how to use Proposition \ref{conlaws} to reduce the number of parameters. Let us write
\begin{equation}\label{int4}
w_1=a e^{-\al},\quad w_2=a e^{\al},\quad w_3=b e^{-\be},\quad w_4=b e^{\be},\quad 
w_5=c e^{-\xi},\quad w_6=c e^{\xi},
\end{equation}
where
\begin{equation}\label{int5}
\begin{aligned}
&a=\sqrt{w_1w_2},\quad e^{\al}=\sqrt{\frac{w_2}{w_1}}\,,\quad
b=\sqrt{w_3w_4},\quad e^{\be}=\sqrt{\frac{w_4}{w_3}}\,,\\
&c=\sqrt{w_5w_6},\quad e^{\xi}=\sqrt{\frac{w_6}{w_5}}\,.
\end{aligned}
\end{equation}
Then
\begin{equation}\label{int6}
w_1^{N_1}w_2^{N_2}w_3^{N_3}w_4^{N_4}w_5^{N_5}w_6^{N_6}
=a^{N_1+N_2} b^{N_3+N_4} c^{N_5+N_6}e^{\al(N_2-N_1)+\be(N_4-N_3)+\xi(N_6-N_5)}
\end{equation}
Let
\begin{equation}\label{int7}
\al=\eta+\theta,\quad \be=\eta-\theta;\qquad \eta=\frac{\al+\be}{2}\,,\quad \theta=\frac{\al-\be}{2}\,.
\end{equation}
Then
\begin{equation}\label{int8}
\begin{aligned}
w_1^{N_1}w_2^{N_2}w_3^{N_3}w_4^{N_4}w_5^{N_5}w_6^{N_6}
&=a^{N_1+N_2} b^{N_3+N_4} c^{N_5+N_6}\\
&\times e^{\eta(N_2-N_1+N_4-N_3)+\theta(N_2-N_1-N_4+N_3)+\xi(N_6-N_5)}
\end{aligned}
\end{equation}
Using the second and third equations of \eqref{int3}, we obtain that
\begin{equation}\label{int9}
\begin{aligned}
w_1^{N_1}w_2^{N_2}w_3^{N_3}w_4^{N_4}w_5^{N_5}w_6^{N_6}
&=a^{N_1+N_2} b^{N_3+N_4} c^{N_5+N_6}
 e^{\eta(N_2-N_1+N_4-N_3)-\theta m(n-m)-\xi(n-m)}\\
&=(a e^{-\eta})^{N_1}(a e^{\eta})^{N_2}(b e^{-\eta})^{N_3}(b e^{\eta})^{N_4}c^{N_5} c^{N_6}e^{-\theta m(n-m)-\xi(n-m)}.
\end{aligned}
\end{equation}
From \eqref{int5}, \eqref{int7},
\begin{equation}\label{int10}
e^{-\theta}=\left(\frac{w_1 w_4  }{w_2 w_3  }\right)^{1/4},\quad e^{-\xi}=\left(\frac{w_5  }{w_6  }\right)^{1/2},
\end{equation}
hence
\begin{equation}\label{int11}
\begin{aligned}
w_1^{N_1}w_2^{N_2}w_3^{N_3}w_4^{N_4}w_5^{N_5}w_6^{N_6}
&=(a e^{-\eta})^{N_1}(a e^{\eta})^{N_2}(b e^{-\eta})^{N_3}(b e^{\eta})^{N_4}c^{N_5} c^{N_6}\\
&\times\left(\frac{w_1 w_4  }{w_2 w_3  }\right)^{m(n-m)/4}\left(\frac{w_5  }{w_6  }\right)^{(n-m)/2}.
\end{aligned}
\end{equation}
This implies the relation between partition functions,
\begin{equation}\label{int12}
\begin{aligned}
Z_{n-m,n}(w_1, w_2, w_3, w_4, w_5, w_6)
&=\left(\frac{w_1 w_4  }{w_2 w_3  }\right)^{m(n-m)/4}\left(\frac{w_5  }{w_6  }\right)^{(n-m)/2}\\\
&\times Z_{n-m,n}(a e^{-\eta},a e^{\eta},b e^{-\eta},b e^{\eta}, c, c),
\end{aligned}
\end{equation}
and between Gibbs measures,
\begin{equation}\label{int13}
\mu(\sg; w_1, w_2, w_3, w_4, w_5, w_6)=\mu(\sg;a e^{-\eta},a e^{\eta},b e^{-\eta},b e^{\eta}, c, c).
\end{equation}
Furthermore, using the first equation of \eqref{int3}, we have
\begin{equation}\label{con5}
\begin{aligned}
Z_{n-m,n}(a e^{-\eta},a e^{\eta},b e^{-\eta},b e^{\eta}, c, c)&=c^{n(n-m)}Z_{n-m,n}\left(\frac{a e^{-\eta}}{c}, \frac{a e^{\eta}}{c},
 \frac{b e^{-\eta}}{c}, \frac{b e^{\eta}}{c}, 1, 1\right), \\ 
\mu(\sg;a e^{-\eta},a e^{\eta},b e^{-\eta},b e^{\eta}, c, c)
&=\mu\left(\sg; \frac{a e^{-\eta}}{c}, \frac{a e^{\eta}}{c},
 \frac{b e^{-\eta}}{c}, \frac{b e^{\eta}}{c}, 1, 1\right),
\end{aligned}
\end{equation}
and so the model reduces to the three parameters, $\frac{a}{c}\,,\;\frac{b}{c}\,,$ and $ \eta$.

\subsection{The main result: asymptotics of the partition function}

Fix two real parameters $t$ and $\ga$, with $0<|\ga|<t$, and introduce the parameterization of $a$, $b$, $c$ as
\begin{equation}\label{in1}
\begin{aligned}
a = \sinh(t-\ga)\,, \quad b =  \sinh(t+\ga)\,, \quad c=\sinh(2|\ga|).
\end{aligned}
\end{equation}
Set
\begin{equation}\label{in1a}
\begin{aligned}
a_{\pm} = e^{\pm \ga}a\,, \qquad b_{\pm} = e^{\pm \ga}b\,.
\end{aligned}
\end{equation}
In the current work we consider the partition function 
\begin{equation}\label{in1b}
Z_{n-m,n}(a_-, a_+, b_-, b_+, c, c)=Z_{n-m,n}(ae^{-\ga}, a e^{\ga}, b e^{-\ga}, b e^{\ga}, c, c),
\end{equation}
depending on the two parameters, $t$ and $\ga$. It is a specialization of the three parameter family,
$Z_{n-m,n}(ae^{-\eta}, a e^{\eta}, b e^{-\eta}, b e^{\eta}, c, c)$ in \eqref{int12} to the case when $\eta=\ga$.

 Notice that $0<\ga<t$ corresponds to the regime
\begin{equation}\label{in03}
b_- b_+ > a_- a_+ +c^2\,,
\end{equation}
and $\ga<0<|\ga|<t$, corresponds to the regime
\begin{equation}\label{in04}
a_- a_+ > b_- b_+ +c^2\,.
\end{equation}
These two regimes are natural extensions of the two components of the {\it ferroelectric regime} considered in \cite{Bleher-Liechty09}.
Without loss of generality we will consider only the component \eqref{in03} corresponding to $0<\ga<t$.

\begin{rem} According to the conservation laws \eqref{int3},
\begin{equation}\label{in5}
\begin{aligned}
Z_{n-m,n}(a_-, a_+, b_-, b_+, c, c)&=e^{-\ga(m(n-m)}Z_{n-m,n}(a, a, e^{-2\ga}b, e^{2\ga}b, c, c) \\
&=e^{-m(n-m)\ga}a^{n(n-m)}Z_{n-m,n}\left(1, 1, \frac{e^{-2\ga}b}{a}, \frac{e^{2\ga}b}{a}, \frac{c}{a}, \frac{c}{a}\right),
\end{aligned}
\end{equation}
where the weights in the latter equation match the weights which appear in \cite{Foda-Wheeler12} after a simple change of variables.
\end{rem}

Our main results concern the asymptotic behavior of the partition function $Z_{n-m,n}$ as $n\to\infty$.

\begin{theo} \label{main} Fix two parameters $t, \ga$ with $0<t<\ga$, and let
\begin{equation}\label{mr1}
 Z_{n-m,n} \equiv Z_{n-m,n}(a_-, a_+, b_-, b_+, c, c),
 \end{equation}
 where $a_\pm, b_\pm,$ and $c$ are as in \eqref{in1}, \eqref{in1a}.
 Fix any $\ep>0$. Then there is a constant $\kappa>0$ such that as $n\to\infty$,
\begin{equation}\label{mr2}
Z_{n-m,n}=[\sinh(t+\ga)]^{n(n-m)} e^{m(n-m)\ga}\,e^{-(n-m)(t-\ga)}\big[1+\mathcal O(e^{-\kappa n})\big],
\end{equation}
uniformly with respect to $m$ in the interval
\begin{equation}\label{mr3}
n\ep\le m<n.
\end{equation}
\end{theo}

{\it Remark.} Observe that there is no constant factor in \eqref{mr2}, i.e., the constant factor is 1. 
In addition, the error term is exponentially small. This should be compared with the paper \cite{Bleher-Liechty09} (see also \cite{Bleher-Liechty14}),
where a similar result was obtained for $m=0$.

The case of $m\le n\ep$, including when $m$ remains bounded, is covered by the following theorem.

\begin{theo} \label{main2} Fix the parameters of the six-vertex models as in Theorem \ref{main}. For any $\ep>0$ there is a constant $n_0>0$ such that for any $n\ge n_0$ and any $0\le m<n$,
\begin{equation}\label{mr4}
Z_{n-m,n}=C(m) [\sinh(t+\ga)]^{n(n-m)} e^{m(n-m)\ga}\,e^{-(n-m)(t-\ga)}\big(1+\xi_{nm}\big),
\end{equation}
where 
\begin{equation}\label{mr5}
C(m)=1-e^{-4\ga(m+1)},
\end{equation}
and
\begin{equation}\label{mr6}
|\xi_{nm}|\le \rho^m e^{-n^{1-\ep}},\quad \rho=e^{-2\ga}<1.
\end{equation}
\end{theo}

The proofs of Theorems \ref{main} and \ref{main2} are based on a determinantal formula 
for $Z_{n-m,n}$ and estimates for corresponding orthogonal polynomials. In fact, Theorem \ref{main} follows 
from Theorem \ref{main2}, but the proof of Theorem \ref{main2} is more involved, and to facilitate the reading
we will first prove Theorem \ref{main}.

It is interesting to notice that the limiting free energy per site $F$ depends on the aspect ratio $r=\frac{n-m}{n}$
of the rectangular lattice of size $(n-m)\times n$. Namely, from \eqref{mr4},
\begin{equation}\label{mr7}
F=\lim_{n,m\to\infty;\; \frac{n-m}{n}\to r}
\frac{1}{(n-m)n} \ln Z_{n-m,n}=\ln[\sinh(t+\ga)]+(1-r)\ga. 
\end{equation}
Indeed, $F$ is determined entirely by the weight of the ground state configuration. Before proceeding with the proof of Theorems \ref{main} and \ref{main2}, let us briefly discuss the ground state.

\subsection{Ground state configuration} The ground state configuration is the one with the largest weight.
For the weights described in Theorem \ref{main}, the weight of a configuration $\sg$ is
\begin{equation}\label{gs1}
\begin{aligned}
w(\sg)&= \left(a_-\right)^{N_1(\sg)}\left(a_+\right)^{N_2(\sg)} \left(b_-\right)^{N_3(\sg)}\left(b_+\right)^{N_4(\sg)}\left(c\right)^{N_5(\sg)+N_6(\sg)} \\
&=e^{-\ga(N_1(\sg)-N_2(\sg))}e^{-\ga(N_3(\sg)-N_4(\sg))}a^{N_1(\sg)+N_2(\sg)}b^{N_3(\sg)+N_4(\sg)}c^{N_5(\sg)+N_6(\sg)}.
\end{aligned}
\end{equation}
Using the conservation laws \eqref{int3}, we can write \eqref{gs1} as
\begin{equation}\label{gs2}
\begin{aligned}
w(\sg)&=e^{- m(n-m)\ga}e^{2\ga(N_4(\sg)-N_3(\sg))}a^{N_1(\sg)+N_2(\sg)}b^{N_3(\sg)+N_4(\sg)}c^{N_5(\sg)+N_6(\sg)}.
\end{aligned}
\end{equation}
Since $b>a+c$, we see that the ground state configuration is the one which maximizes both $N_4+N_3$ and $N_4-N_3$.  This is achieved by placing type-5 vertices along the up-left diagonal starting from the bottom-right corner. Above this diagonal all arrows point down or right, and so all vertices are of type 3. Below and to the left of this diagonal, all arrows point up or left, and so all vertices are of type 4. See Fig. \ref{ground_state} for the ground state configuration on the $4\times 7$ lattice. The diagonal of type-5 vertices is circled.
\begin{figure}
\begin{center}\scalebox{0.45}{\includegraphics{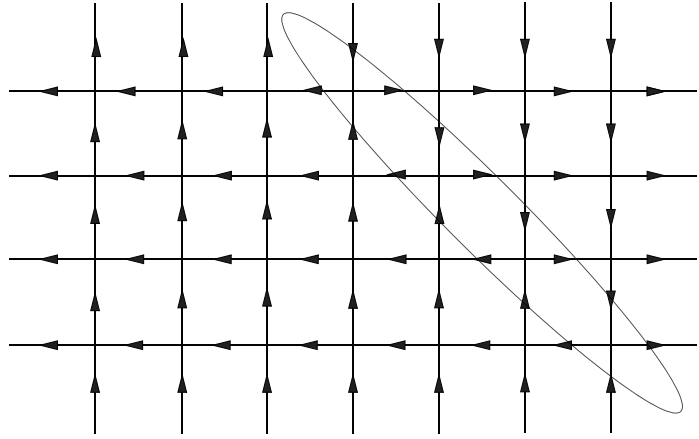}}\end{center}
\caption{The ground state configuration on the $4\times 7$ lattice. On the circled diagonal all vertices are of type 5. To the right of this diagonal all vertices are of type 3, and to the left of it all vertices are of type 4.}
\label{ground_state}
\end{figure}

The weight of the ground state configuration $\sg^{\textrm{gs}}$  is
\begin{equation}\label{gs3}
\begin{aligned}
w(\sg^{\textrm{gs}})&=\left(b_+\right)^{(n-m)(n+m-1)/2}\left(b_-\right)^{(n-m)(n-m-1)/2}c^{n-m} \\
&=e^{m(n-m)\ga} b^{(n-m)(n-1)} c^{n-m} \\
&=b^{n(n-m)} e^{m(n-m)\ga} \left(\frac{c}{b}\right)^{n-m}\,.
\end{aligned}
\end{equation}
Comparing \eqref{mr2} and \eqref{mr4} with \eqref{gs3} we find that as $m,n\to \infty$,
\begin{equation}\label{gs4}
\frac{Z_{n-m,n}}{w(\sg^{\textrm{gs}})}= \left(\frac{e^{2\ga}-e^{-2t}}{e^{2\ga}-e^{-2\ga}}\right)^{n-m}(1+\xi_{nm}),
\end{equation}
where  $\xi_{nm}$ is estimated in Theorem \ref{main2}.
A comparison of \eqref{mr2} or \eqref{mr4} with \eqref{gs3} shows that the limiting free energy per site $F$ defined in \eqref{mr7}  comes entirely from the ground state configuration. This was known for $m=0$, see \cite{Bleher-Liechty09}.

\subsection{Outline of the rest of the paper}
The rest of the paper is organized as follows. In section \ref{opform} we state the determinantal formula for the partition function $Z_{n-m,n}$, and use it to write a formula for $Z_{n-m,n}$ in terms of certain orthogonal polynomials on the positive integer lattice. In section \ref{Meixner} we recall the Meixner polynomials, and in section \ref{RHA} we introduce the Interpolation Problem in order to compare our orthogonal polynomials with the Meixner ones. In sections \ref{hkk} and \ref{proofmain}, we prove Theorem \ref{main} by a careful comparison of the normalizing constants for our orthogonal polynomials with those of the Meixner polynomials. In section \ref{hkk0} the analogous analysis is carried out for the proof of Theorem \ref{main2}, and in sections \ref{evaluation_C} and \ref{formula_C} the proof of Theorem \ref{main2} is completed by finding an explicit formula for the constant term $C(m)$ using the Toda equation. Finally, section \ref{phase} gives a short discussion of the phase transition in the underlying orthogonal polynomials.

\section{An orthogonal polynomial formula for $Z_{n-m,n}$}\label{opform}
 
Introduce the notations
\begin{equation}\label{in3}
\begin{aligned}
\varphi(t):&= \sinh(t-\ga)\sinh(t+\ga)=ab\,,\\  \phi(t):&=\frac{\sinh(2\ga)}{\sinh(t-\ga)\sinh(t+\ga)}=\frac{c}{ab}\,. \\
\end{aligned}
\end{equation}
The starting point for our analysis is the following determinantal formula for the partition function.
\begin{prop}\label{part_function_det_formula}
For the six-vertex model on the $(n-m)\times n$ lattice with pDWBC and weights as described in \eqref{in1} -- \eqref{in1a}, the partition function $Z_{n-m,n}$ is given by the following formula:
\begin{equation}\label{in4}
Z_{n-m,n} = \frac{(-1)^{m(m+1)/2-nm}\varphi(t)^{n(n-m)} e^{m(n-m)t} }{2^{m(m-1)/2}\prod_{j=0}^{n-m-1}j! \prod_{j=0}^{n-1}j! }\;
 \tau_{n-m,n}, \\
\end{equation}
where 
\begin{equation}\label{in4a}
\tau_{n-m,n} := \det \begin{pmatrix} 1 & (-2) & (-2)^2 & \dots & (-2)^{n-1} \\ \vdots & \vdots &\vdots & \ddots & \vdots \\ 1 & (-2m) & (-2m)^2 & \dots & (-2m)^{n-1} \\ \phi(t) & \phi'(t) & \phi''(t) & \dots & \phi^{(n-1)}(t)  \\ \phi'(t) & \phi''(t) & \phi'''(t) & \dots & \phi^{(n)}(t) \\
\vdots & \vdots &\vdots & \ddots & \vdots \\  \phi^{(n-m-1)}(t) & \phi^{(n-m)}(t) & \phi^{(n-m+1)}(t) & \dots & \phi^{(2n-m-2)}(t) \end{pmatrix},
\end{equation}
is a determinant of mixed Vandermonde/Hankel type.
\end{prop}
This proposition follows from the Izergin--Korepin formula for the partition function of the inhomogeneous six-vertex model with domain wall boundary conditions on the square lattice \cite{Izergin87}, \cite{Izergin-Coker-Korepin92}. A similar formula for the inhomogeneous six-vertex model was derived in \cite{Foda-Wheeler12}. For completeness, we present a proof in Appendix \ref{appendix1}. When $m=0$, it simplifies to the
usual formula for the partition function of the six-vertex model with domain wall boundary conditions: 
\begin{equation}\label{in4b}
Z_{n,n} = \frac{[\varphi(t)]^{n^2}  }{ \prod_{j=0}^{n-1}(j!)^2 }\;
 \tau_{n,n}, \qquad \tau_{n,n}=\det\big( \phi^{j+k-2}(t)\big)_{j,k=1}^n.
\end{equation}
\begin{rem} The proof of Proposition \ref{part_function_det_formula} relies on the particular parametrization \eqref{in1} -- \eqref{in1a} of the weights. For the general class of weights \eqref{int12}, there doesn't seem to be any determinantal formula for the partition function unless $\eta=\ga$. See the remark following the proof of lemma \ref{inductive_lemma}
\end{rem}

It is straightforward that for $0<\ga<t$, the function $\phi(t)$ can be represented as the discrete Laplace transform
 \begin{equation}\label{as1}
 \phi(t)= 2\sum_{x=1}^\infty u(x)\,, \qquad u(x):=2e^{-2tx} \sinh(2\ga x)=e^{-2(t-\ga)x}-e^{-2(t+\ga)x}\,.
 \end{equation}
Indeed,
\begin{equation}\label{as2}
\begin{aligned}
2\sum_{x=1}^\infty u(x)&=2\sum_{x=1}^\infty  \left(e^{-2(t-\ga)x}-e^{-2(t+\ga)x}\right)=2\left[
\frac{e^{-2(t-\ga)}}{1-e^{-2(t-\ga)}}-\frac{e^{-2(t+\ga)}}{1-e^{-2(t+\ga)}}\right]\\
&=\frac{2\left[e^{-2(t-\ga)}-e^{-2(t+\ga)}\right]}
{\left[1-e^{-2(t-\ga)}\right]\,\left[1-e^{-2(t+\ga)}\right]}
=\frac{\sinh (2\ga)}{\sinh(t-\ga)\sinh(t+\ga)}=\phi(t)\,.
\end{aligned}
 \end{equation}
From \eqref{as1}, 
 \begin{equation}\label{as3}
 \phi^{(k)}(t)=2 \sum_{x=1}^\infty (-2x)^k u(x)\,.
 \end{equation}
and by multi-linearity of the determinant, we obtain from \eqref{in4a} that
\begin{equation}\label{as4}
\begin{aligned}
\tau_{n-m,n} &= 2^{n-m} \sum_{x_1, \dots, x_{n-m}=1}^\infty \det \begin{pmatrix} 1 & (-2)  & \dots & (-2)^{n-1}  \\ 1 & (-4) & \dots &  (-4)^{n-1}  \\
\vdots & \vdots & \ddots & \vdots \\
1 & (-2m)  & \dots & (-2m)^{n-1} \\
1 & -2x_{1} & \dots & (-2x_{1})^{n-1} \\
-2x_{2} & (-2x_{2})^2 & \dots & (-2x_{2})^n \\
\vdots & \vdots & \ddots & \vdots \\
(-2x_{n-m})^{n-m-1} & (-2x_{n-m})^{n-m} & \dots & (-2x_{n-m})^{2n-m-2} 
\end{pmatrix} \\
&\hspace{11 cm}\times\prod_{j=1}^{n-m} u(x_j) \\
&= 2^{n-m}(-2)^{(n-m-1)(n-m)/2}(-2)^{n(n-1)/2} \sum_{x_1, \dots, x_{n-m}=1}^\infty \det \begin{pmatrix} 1 &  1  & \dots & 1  \\ 1 &  2 & \dots &  2^{n-1}  \\
\vdots & \vdots & \ddots & \vdots \\
1 & m  & \dots & m^{n-1} \\
1 & x_{1} & \dots & x_{1}^{n-1} \\
1 & x_{2} & \dots & x_{2}^{n-1} \\
\vdots & \vdots & \ddots & \vdots \\
1 & x_{n-m} & \dots & x_{n-m}^{n-1}
\end{pmatrix} \\
&\hspace{9 cm} \times x_1^0 x_2^1 \cdots x_{n-m}^{n-m-1} \prod_{j=1}^{n-m} u(x_j) \\
&= (-1)^{m(m+1)/2-nm}2^{n(n-m)+m(m-1)/2}  \sum_{x_1, \dots, x_{n-m}=1}^\infty \De_n(\mathbf y) \prod_{j=1}^{n-m} x_j^{j-1}u(x_j).
\end{aligned}
\end{equation}
where $\De_n(\mathbf y)=\prod_{1\le j<k\le n}(y_k-y_j)$ is the $n$-dimensional Vandermonde determinant, and $\mathbf y=(y_1, y_2, \dots, y_n)$ is an $n$-dimensional vector whose first $m$ components are the first $m$ natural numbers, and whose remaining $n-m$ components are the summation variables $x_j$:
  \begin{equation}\label{as5}
     \begin{aligned}
     y_j&=j, \qquad &j=1,2,\dots, m, \\
y_{m+j}&=x_j\,, \qquad &j=1,2,\dots, n-m.
\end{aligned}
\end{equation}
Define the vector $\mathbf A:=(1,2,3,\dots, m)$, and introduce the function
  \begin{equation}\label{as6}
  g_m(x):= \prod_{k=1}^m (x-k)\,.
  \end{equation}
Observe that the Vandermonde determinant $\De_n(\mathbf y)$ can be factored as
  \begin{equation}\label{as7}
  \De_n(\mathbf y) = \De_{n-m}(\mathbf x) \De_m(\mathbf A) \prod_{j=1}^{n-m} g_m(x_j)\,,
  \end{equation}
where $\De_{n-m}$ and $\De_m$ are the $(n-m)$- and $m$-dimensional Vandermonde determinants, respectively, and $\mathbf x=(x_1, \dots, x_{n-m})$.
We thus can write \eqref{as4} as
\begin{equation}\label{as8}
  \tau_{n-m,n}=(-1)^{m(m+1)/2-nm}2^{n(n-m)+m(m-1)/2}\De_m(\mathbf A)\upsilon_{m,n-m} \,,
    \end{equation}
 where
 \begin{equation}\label{as9}
  \upsilon_{m,n-m}= \sum_{x_1, \dots, x_{n-m}=1}^\infty \De_{n-m}(\mathbf x) \left(\prod_{j=1}^{n-m} x_j^{j-1}g_m(x_j)u(x_j)\right)\,.
  \end{equation}
A standard symmetrization argument then gives
  \begin{equation}\label{as10}
  \upsilon_{m,n-m}=\frac{1}{(n-m)!}\sum_{x_1, \dots, x_{n-m}=1}^\infty \De_{n-m}(\mathbf x)^2 \left(\prod_{j=1}^{n-m} g_m(x_j)u(x_j)\right) \,.
    \end{equation}
Since the function $g_m(x)$ vanishes for $x=1, 2, \dots, m$, this sum is in fact
  \begin{equation}\label{as11}
  \upsilon_{m,n-m}=\frac{1}{(n-m)!}\sum_{x_1, \dots, x_{n-m}=m+1}^\infty \De_{n-m}(\mathbf x)^2 \left(\prod_{j=1}^{n-m} g_m(x_j)u(x_j)\right) \,.
    \end{equation}
It is convenient to shift $x_j$'s by $m+1$:
  \begin{equation}\label{as12}
  \upsilon_{m,n-m}=\frac{1}{(n-m)!}\sum_{x_1, \dots, x_{n-m}=0}^\infty \De_{n-m}(\mathbf x)^2 \prod_{j=1}^{n-m} w(x_j)\,.
    \end{equation}
where
  \begin{equation}\label{as13}
\begin{aligned}
w(x)&=g_m(x+m+1)u(x+m+1)\\
&=\left[e^{-2(t-\ga)(x+m+1)}-e^{-2(t+\ga)(x+m+1)}\right]\prod_{k=1}^m (x+m+1-k).
\end{aligned}
    \end{equation}
The last product can be rearranged as follows:
  \begin{equation}\label{as14}
w(x)=\left[e^{-2(t-\ga)(x+m+1)}-e^{-2(t+\ga)(x+m+1)}\right]\prod_{k=1}^m (x+k).
    \end{equation}
Notice that we assume that $t>\ga>0$,
hence $w(x)>0$ for $x\ge 0$. Therefore we can introduce monic polynomials orthogonal with respect to the  weight $w(x)$:
    \begin{equation}\label{as15}
\sum_{x =0}^\infty p_j(x) p_k(x) w(x)= h_{k} \de_{jk}\,,
    \end{equation}
  where $h_{k}=h_k(m)  >0$ are normalizing constants.  We then have
      \begin{equation}\label{as16}
  \upsilon_{m,n-m}= \prod_{k=0}^{n-m-1} h_{k}.
    \end{equation}
Observe that
      \begin{equation}\label{as17}
 \De_m(\mathbf A)= \prod_{1\le j<k\le m}(k-j)=\prod_{j=0}^{m-1} j!\,,
    \end{equation}
hence by \eqref{as8},
      \begin{equation}\label{as18}
  \tau_{n-m,n}=(-1)^{m(m+1)/2-nm}2^{n(n-m)+m(m-1)/2} \left(\prod_{j=0}^{m-1} j! \right)\upsilon_{m,n-m}.
    \end{equation}
Thus, from \eqref{in4} and \eqref{as16} we obtain the following formula for the partition function:
\begin{equation}\label{as19}
\begin{aligned}
Z_{n-m,n}
= (2ab)^{n(n-m)} e^{m(n-m)t}\left(\frac{\prod_{j=0}^{m-1} j!}{\prod_{j=0}^{n-m-1} j!\prod_{j=0}^{n-1} j!} \right) 
 \prod_{j=0}^{n-m-1} h_{j}\,,
\end{aligned}
\end{equation}
where
\begin{equation}\label{as20}
a=\sinh(t-\ga),\qquad b=\sinh(t+\ga).
\end{equation}
Equivalently, it can be written as follows:

\begin{prop}
\begin{equation}\label{as21}
\begin{aligned}
Z_{n-m,n}
= (2ab)^{n(n-m)} e^{m(n-m)t}\prod_{j=0}^{n-m-1}\frac{h_j}{j!\,(j+m)!} \,.
\end{aligned}
\end{equation}
\end{prop}

\section{Approximation by the Meixner polynomials}\label{Meixner}

Let us rewrite formula \eqref{as14} as
  \begin{equation}\label{meix1}
w(x)=\left[e^{-2(t-\ga)(x+m+1)}-e^{-2(t+\ga)(x+m+1)}\right]\frac{(x+m)!}{x!}\,,
    \end{equation}
or
  \begin{equation}\label{meix2}
w(x)
=\left[e^{-2(t-\ga)(x+m+1)}-e^{-2(t+\ga)(x+m+1)}\right]\frac{m!(m+1)_x}{x!}\,,
    \end{equation}
where
  \begin{equation}\label{meix3}
(\be)_x=\be(\be+1)\ldots (\be+x-1),
    \end{equation}
is the Pochhammer symbol. As an approximation to $w(x)$, let us consider the weight
  \begin{equation}\label{meix4}
\begin{aligned}
&w^{\rm M}(x)=e^{-2(t-\ga)(x+m+1)}\,\frac{m!(m+1)_x}{x!}
=C_m e^{-2(t-\ga)x}\frac{(m+1)_x}{x!}\,,\\
& C_m=m!e^{-2(t-\ga)(m+1)},
\end{aligned}
    \end{equation}
so that
  \begin{equation}\label{meix5}
w(x)
=w^{\rm M}(x)\left[ 1-e^{-4\ga(x+m+1)}\right]\,.
    \end{equation}
The orthogonal polynomials with respect to the weight $w^{\rm M}(x)$ are the {\it Meixner polynomials}.

The Meixner polynomials $M_k(z;\be,q)$ with parameters $\be>0$ and $0<q<1$ are defined as
\begin{equation} \label{meix6}
\begin{aligned}
M_k(z;\be,q)&= {}_2F_1\left(\begin{matrix} -k,-z \\ \be \end{matrix};1-q^{-1}\right)
=\sum_{j=0}^k \frac{(-k)_j(-z)_j (1-q^{-1})^j}{(\be)_j j!}\,\\
&=\sum_{j=0}^{k}\frac{(1-q^{-1})^j \di\prod_{i=0}^{j-1}(k-i) \prod_{i=0}^{j-1}(z-i)}{(\be)_j j!}\,.
\end{aligned}
\end{equation}
They satisfy the orthogonality condition,
\begin{equation} \label{meix7}
\sum_{x=0}^\infty  M_j(x;\be,q)M_k(x;\be,q)\,\frac{(\be)_x q^x}{x!}= \frac{k!\,\de_{jk}}{(\be)_k q^k(1-q)^\be}\,,
\end{equation}
see, e.g. \cite{Koekoek-Lesky-Swarttouw10}. By \eqref{meix6}, the leading coefficient in the Meixner polynomial $M_k(z;\be,q)$ is
\begin{equation} \label{meix8}
\begin{aligned}
M_k(z;\be,q)=\frac{(1-q^{-1})^k  }{(\be)_k }\,z^k+\ldots.
\end{aligned}
\end{equation}
For the corresponding monic polynomials,
\begin{equation} \label{meix9}
p_k^{\rm M}(z)=\frac{(\be)_k}{(1-q^{-1})^k} M_k(z;\be,q),
\end{equation}
(M in $p_k^{\rm M}$ stands for Meixner), the orthogonality condition reads 
\begin{equation} \label{meix10}
\sum_{x=0}^\infty p_j^{\rm M}(x)p_k^{\rm M}(x)\,\frac{(\be)_x q^x}{x!}=\frac{(\be)_k q^k k!\,\de_{jk}}{(1-q)^{\be+2k}}\,.
\end{equation}
To relate it to the weight $w^{\rm M}(x)$ in \eqref{meix4}, we set 
\begin{equation} \label{meix11}
\be=m+1,\quad q=e^{-2(t-\ga)}.
\end{equation}
Then \eqref{meix4}, \eqref{meix10} imply that
\begin{equation} \label{meix12}
\begin{aligned}
&w^{\rm M}(x)=C_mq^x\,\frac{(m+1)_x}{x!}\,,\quad C_m=m!q^{m+1},\\
&\sum_{x=0}^\infty p_j^{\rm M}(x)p_k^{\rm M}(x)\,w^{\rm M}(x)= h^{\rm M}_k\de_{jk}\,,
\quad h^{\rm M}_k=\frac{k!\,(k+m)!\,q^{k+m+1}}{(1-q)^{2k+m+1}}.
\end{aligned}
\end{equation}
As an approximation to the partition function $Z_{n-m,n}$ in \eqref{as19}, we introduce the {\it Meixner partition function},
\begin{equation}\label{meix13}
\begin{aligned}
Z_{n-m,n}^{\rm M}
= (2ab)^{n(n-m)} e^{m(n-m)t}\prod_{j=0}^{n-m-1} \frac{h_j^{\rm M}}{j!\,(j+m)!}\,.
\end{aligned}
\end{equation}
\begin{rem}
The factor $\prod_{j=0}^{n-m-1} h_j^{\rm M}$ appearing in the Meixner partition function is identical to the normalizing constant in a particular expression for the last passage time in the point-to-point last passage percolation model on a rectangular lattice with geometric weights, see \cite[Proposition 1.3]{Johansson00}.
\end{rem}

From \eqref{meix12} we obtain that
\begin{equation}\label{meix14}
\begin{aligned}
\prod_{j=0}^{n-m-1} \frac{h_j^{\rm M}}{j!\,(j+m)!}
=\prod_{j=0}^{n-m-1} \frac{q^{j+m+1}}{(1-q)^{2j+m+1}}
=\frac{q^{(n+m+1)(n-m)/2}}{(1-q)^{n(n-m)}}\,,
\end{aligned}
\end{equation}
hence
\begin{equation}\label{meix15}
\begin{aligned}
Z_{n-m,n}^{\rm M}
= (2ab)^{n(n-m)} e^{m(n-m)t}\,\frac{q^{(n+m+1)(n-m)/2}}{(1-q)^{n(n-m)}}\,.
\end{aligned}
\end{equation}
By \eqref{as20},
\begin{equation}\label{meix16}
2a=e^{t-\ga}-e^{-(t-\ga)}=\frac{1-q}{q^{1/2}}\,.
\end{equation}
Substituting this into \eqref{meix15} and simplifying, we obtain that
\begin{equation}\label{meix17}
\begin{aligned}
Z_{n-m,n}^{\rm M}
&= b^{n(n-m)} e^{m(n-m)t}\,q^{(m+1)(n-m)/2}
= b^{n(n-m)} e^{m(n-m)t}\,e^{-(m+1)(n-m)(t-\ga)}\\
&= b^{n(n-m)} e^{m(n-m)\ga}\,e^{-(n-m)(t-\ga)}\,.
\end{aligned}
\end{equation}
Now we would like to estimate the ratio,
\begin{equation}\label{meix18}
\begin{aligned}
\frac{Z_{n-m,n}}{Z_{n-m,n}^{\rm M}}
=\prod_{k=0}^{n-m-1} \frac{h_k}{h^{\rm M}_k}\,.
\end{aligned}
\end{equation}
This will be done in the sections \ref{hkk} and \ref{hkk0} by showing that $h_k/h_k^{\rm M}$ is exponentially close to 1 as  $k\to\infty$. As a means to compare the two systems of orthogonal polynomials, let us first introduce the Interpolation Problem for each system.

\section{Riemann Hilbert approach: Interpolation problem}\label{RHA}

The Riemann-Hilbert approach to discrete orthogonal polynomials
 is based on the following Interpolation Problem (IP), which was
introduced in the paper \cite{Borodin-Boyarchenko03} of Borodin and Boyarchenko
under the name of the discrete Riemann-Hilbert problem.
See also the monograph  \cite{BKMM}
of Baik, Kriecherbauer, McLaughlin, and Miller, in which it is called the
Interpolation Problem.  
Let $w(l)\ge 0$ be a weight function on $\Z_+=\{l=0,1,2,\ldots\}$ (it can be a more general
discrete set, as discussed in \cite{Borodin-Boyarchenko03} and \cite {BKMM}, but we will
need $\Z_+$ in our problem).

{\bf Interpolation Problem}. For a given $k=0,1,\ldots$, find a $2\times 2$ matrix-valued function
$\mathbf P(z;k)=(\mathbf P_{ij}(z;k))_{1\le i,j\le 2}$ with the following properties:
\begin{enumerate}
\item
{\bf Analyticity}: $\mathbf P(z;k)$ is an analytic function of $z$ for $z\in\C\setminus\Z_+$.
\item
{\bf Residues at poles}: At each node $l\in\Z_+$, the elements $\mathbf P_{11}(z;k)$ and
$\mathbf P_{21}(z;k)$ of the matrix $\mathbf P(z;k)$ are analytic functions of $z$, and the elements $\mathbf P_{12}(z;k)$ and
$\mathbf P_{22}(z;k)$ have a simple pole with the residues,
\begin{equation} \label{IP1}
\underset{z=l}{\rm Res}\; \mathbf P_{j2}(z;k)=w(l)\mathbf P_{j1}(l;k),\quad j=1,2.
\end{equation}
Equivalently, the latter relation can be written in the matrix form as
\begin{equation} \label{IP2}
\underset{z=l}{\rm Res}\; \mathbf P(z;k)=
\mathbf P(l;k)
\begin{pmatrix}
0 & w(l)\\
0 & 0
\end{pmatrix}.
\end{equation}
\item
{\bf Asymptotics at infinity}: As $z\to\infty$, $\mathbf P(z;k)$ admits the asymptotic expansion,
\begin{equation} \label{IP3}
\mathbf P(z;k)\sim \left(\mathbf  I+\frac {\mathbf P_1}{z}+\frac {\mathbf P_2}{z^2}+\ldots\right)
\begin{pmatrix}
z^k & 0 \\
0 & z^{-k}
\end{pmatrix},\qquad z\in \C\setminus \left[\bigcup_{l=0}^\infty D(l,r_l)\right],
\end{equation}
where $D(z,r)$ is a disk of radius $r>0$ centered at $z\in \C$ and 
\begin{equation} \label{IP4}
\lim_{l\to\infty} r_l=0.
\end{equation}
\end{enumerate}

It is not difficult to see (see \cite{Borodin-Boyarchenko03} and \cite{BKMM}) that under some mild conditions on $w(l)$,
the IP has a unique solution, which is
\begin{equation} \label{IP5}
\mathbf P(z;k)=
\begin{pmatrix}
p_k(z) & C(wp_k)(z) \\
(h_{k-1})^{-1}p_{k-1}(z) & (h_{k-1})^{-1}C(wp_{k-1})(z)
\end{pmatrix}
\end{equation}
where the discrete Cauchy transformation $C$ is defined by the formula,
\begin{equation} \label{IP6}
C(f)(z)=\sum_{l=0}^\infty\frac{f(l)}{z-l}\,,
\end{equation}
and $p_k(z)=z^k+\ldots$ are monic polynomials orthogonal with the weight $w(l)$,
so that
\begin{equation} \label{IP7}
\sum_{l=0}^\infty p_j(l)p_k(l)w(l)=h_j\delta_{jk}.
\end{equation}
It follows from \eqref{IP5} that 
\begin{equation} \label{IP8}
h_k=[\mathbf P_1]_{21},
\end{equation}
where $[\mathbf P_1]_{21}$ is the (21)-element of the matrix $\mathbf P_1$ on the right in (\ref{IP3}).
In what follows we will consider the solution
$\mathbf P(z;k)$ for the weight $w$, introduced in (\ref{meix2}). 

In principle we could apply the nonlinear steepest descent method of Deift and Zhou to this Interpolation Problem to obtain asymptotic expressions for the normalizing constants $h_k$ as $k\to\infty$. This analysis is very similar to the steepest descent analysis for the Meixner polynomials which was carried out by Wang and Wong \cite{Wang-Wong11}, although they considered the parameter $\be$ in \eqref{meix10} to be fixed, while we allow it to grow with $k$. In this paper we take a different approach and compare the normalizing constants $h_k$ with the Meixner normalizing constants $h_{k}^{\rm M}$, for which we have the exact formulae \eqref{meix12}. In order to compare them, it is convenient to also introduce the Riemann-Hilbert problem for the Meixner polynomials.

Let $\mathbf P^{\rm M}$ be  a solution to the IP with the weight $w^{\rm M}$,
\begin{equation} \label{IP9}
\mathbf P^{\rm M}(z;k)=
\begin{pmatrix}
p^{\rm M}_k(z) & C(w^{\rm M}p^{\rm M}_k)(z) \\
(h_{k-1}^{\rm M})^{-1}p^{\rm M}_{k-1}(z) & (h_{k-1}^{\rm M})^{-1}C(w^{\rm M}p^{\rm M}_{k-1})(z)
\end{pmatrix}.
\end{equation}
Consider the quotient matrix,
\begin{equation} \label{IP10}
\mathbf X(z;k)=\mathbf P(z;k)[\mathbf P^{\rm M}(z;k)]^{-1}.
\end{equation}
Observe that $\det \mathbf P^{\rm M}(z;k)$ has no poles and it approaches 1 as $z\to \infty $
outside of the disks $D(l,r_l)$, $l=1,2,\ldots$, hence
\begin{equation} \label{IP11}
\det \mathbf P^{\rm M}(z;k)=1
\end{equation}
and
\begin{equation} \label{IP12}
[\mathbf P^{\rm M}(z;k)]^{-1}=
\begin{pmatrix}
(h_{k-1}^{\rm M})^{-1}C(w^{\rm M}p^{\rm M}_{k-1})(z) & -C(w^{\rm M}p^{\rm M}_k)(z) \\
-(h_{k-1}^{\rm M})^{-1}p^{\rm M}_{k-1}(z) & p^{\rm M}_k(z) 
\end{pmatrix}.
\end{equation}
The matrix-valued function $\mathbf X(z;k)$ solves the following IP:

{\bf Interpolation Problem for $\mathbf X(z;k)$}. 
\begin{enumerate}
\item
{\bf Analyticity}: $\mathbf X(z;k)$ is an analytic function of $z$ for $z\in\C\setminus\Z_+$.
\item
{\bf Residues at poles}: At each node $l\in\Z_+$, 
\begin{equation} \label{IP13}
\underset{z=l}{\rm Res}\; \mathbf X(z;k)=
\mathbf X(l;k) \mathbf J_X(l;k),
\end{equation}
where
\begin{equation} \label{IP14}
\begin{aligned}
 \mathbf J_X(l;k)&=\mathbf P^{\rm M}(l;k)
\begin{pmatrix}
0 & w(l)-w^{\rm M}(l)\\
0 & 0
\end{pmatrix}
[\mathbf P^{\rm M}(l;k)]^{-1}\\
&=[w(l)-w^{\rm M}(l)]
\begin{pmatrix}
-(h_{k-1}^{\rm M})^{-1}p^{\rm M}_{k-1}(l)p^{\rm M}_k(l)  & [p^{\rm M}_k(l)]^2 \\
[(h_{k-1}^{\rm M})^{-1}p^{\rm M}_{k-1}(l)]^2 & (h_{k-1}^{\rm M})^{-1}p^{\rm M}_{k-1}(l)p^{\rm M}_k(l) 
\end{pmatrix}.
\end{aligned}
\end{equation}
\item
{\bf Asymptotics at infinity}: As $z\to\infty$, $\mathbf X(z;k)$ admits the asymptotic expansion,
\begin{equation} \label{IP15}
\begin{aligned}
\mathbf X(z;k)&\sim \left(\mathbf  I+\frac {\mathbf X_1}{z}+\frac {\mathbf X_2}{z^2}+\ldots\right),
\quad z\in \C\setminus \left[\bigcup_{l=0}^\infty D(l,r_l)\right].
\end{aligned}
\end{equation}
\end{enumerate}

From \eqref{IP10} we obtain that in \eqref{IP15}
\begin{equation} \label{IP16}
\begin{aligned}
 \mathbf I+\frac {\mathbf X_1}{z}+\frac {\mathbf X_2}{z^2}+\ldots
=\left(\mathbf  I+\frac {\mathbf P_1}{z}+\frac {\mathbf P_2}{z^2}+\ldots\right)
\left(\mathbf  I+\frac {\mathbf P^{\rm M}_1}{z}+\frac {\mathbf P^{\rm M}_2}{z^2}+\ldots\right)^{-1},
\end{aligned}
\end{equation}
where on the right hand side we use a formal multiplication and inversion of power series in $1/z$.
In particular,
\begin{equation} \label{IP17}
\mathbf X_1=\mathbf P_1-\mathbf P^{\rm M}_1,
\end{equation}
hence by \eqref{IP8},
\begin{equation} \label{IP18}
[\mathbf X_1]_{12}=h_k-h^{\rm M}_k.
\end{equation}
It is easy to check that the matrix 
\begin{equation} \label{IP19}
\mathbf X(z;k)=\mathbf I +C[(w^{\rm M}-w)\mathbf R](z;k),
\end{equation}
where
\begin{equation} \label{IP20}
\mathbf R(z;k)=
\begin{pmatrix}
(h_{k-1}^{\rm M})^{-1}p_k(z) p^{\rm M}_{k-1}(z)  & -p_k(z)p^{\rm M}_k(z) \\
(h_{k-1} h_{k-1}^{\rm M})^{-1}p_{k-1}(z) p^{\rm M}_{k-1}(z) & \;-(h_{k-1})^{-1}p_{k-1}(z)p^{\rm M}_k(z) 
\end{pmatrix},
\end{equation}
solves the IP for $\mathbf X(z;k)$. The uniqueness of the solution of the IP implies that $\mathbf X(z;k)$
is given by formula \eqref{IP19}.

From (\ref{IP19}) and (\ref{IP20}) we obtain that
\begin{equation} \label{IP21}
h_k-h_k^{\rm M}=\sum_{l=0}^\infty p_k(l)p^{\rm M}_k(l)\,[w(l)-w^{\rm M}(l)].
\end{equation}
We will use this identity to estimate $|h_k-h_k^{\rm M}|$.

We would like to remark that identity (\ref{IP21}) can be also derived as follows.
Observe that since $p_k$ and $p^{\rm M}_k$ are monic polynomials, the difference, $p_k-p^{\rm M}_k$,
is a polynomial of degree less than $k$, hence
\begin{equation} \label{IP22}
\sum_{l=0}^\infty p_k(l)[p_k(l)-p^{\rm M}_k(l)]w(l)=0.
\end{equation}
By adding this to equation (\ref{IP7}) with $j=k$, we obtain that
\begin{equation} \label{IP23}
h_k=\sum_{l=0}^\infty p_k(l)p^{\rm M}_k(l)w(l).
\end{equation}
Similarly we obtain that
\begin{equation} \label{IP24}
h_k^{\rm M}=\sum_{l=0}^\infty p_k(l)p^{\rm M}_k(l)w^{\rm M}(l).
\end{equation} 
By subtracting the last two equations, we obtain identity (\ref{IP21}).

\section{Evaluation of the ratio $h_k/h_k^{\rm M}$ for $n\ep\le m<n$}\label{hkk}

In this section we prove the following result:

\begin{prop} Fix any $\ep>0$. Then there is a constant $\kappa>0$ such that
\begin{equation} \label{h1}
h_k=h_k^{\rm M}e^{r_k},
\end{equation}
where 
\begin{equation} \label{h2}
r_k =\mathcal O (e^{-\kappa n}),\quad k=0,1,2,\ldots,
\end{equation}
uniformly with respect to $m$ in the interval $n\ep\le m<n$ and $k\in \Z_+$.
\end{prop}

\begin{proof} Applying the Cauchy-Schwarz inequality to identity (\ref{IP21}), we obtain that
\begin{equation} \label{h3}
|h_k-h_k^{\rm M}|\le\left[\sum_{l=0}^\infty [p_k(l)]^2\,|w(l)-w^{\rm M}(l)|\right]^{1/2}
\left[\sum_{l=0}^\infty [p^{\rm M}_k(l)]^2\,|w(l)-w^{\rm M}(l)|\right]^{1/2},
\end{equation}
which implies that
\begin{equation} \label{h4}
\begin{aligned}
\left|\left(\frac{h_k}{h_k^{\rm M}}\right)^{1/2}-\left(\frac{h_k^{\rm M}}{h_k}\right)^{1/2}\right|
&\le\left[\frac{1}{h_k}
\sum_{l=0}^\infty [p_k(l)]^2\,|w(l)-w^{\rm M}(l)|\right]^{1/2}\\
&\times \left[\frac{1}{h_k^{\rm M}}\sum_{l=0}^\infty [p^{\rm M}_k(l)]^2\,|w(l)-w^{\rm M}(l)|\right]^{1/2},
\end{aligned}
\end{equation}
From (\ref{meix5}),
\begin{equation} \label{h5}
\begin{aligned}
&|w(l)-w^{\rm M}(l)|= \frac{w(l)}{e^{4\ga(l+m+1)}-1}\le C_0 w(l),\quad l\ge 0,\quad C_0=\frac{1}{e^{4\ga(m+1)}-1}\,,\\
&|w(l)-w^{\rm M}(l)|= w^{\rm M}(l)\,e^{-4\ga(l+m+1)}\le C_1 w^{\rm M}(l),\quad l\ge 0,\quad C_1=e^{-4\ga(m+1)}\,,
\end{aligned}
\end{equation}
hence
\begin{equation} \label{h6}
\begin{aligned}
&\frac{1}{h_k}\sum_{l=0}^\infty [p_k(l)]^2\,|w(l)-w^{\rm M}(l)|
\le C_0\frac{1}{h_k}\sum_{l=0}^\infty [p_k(l)]^2 w(l)= C_0,\\
&\frac{1}{h_k^{\rm M}}\sum_{l=0}^\infty [p_k^{\rm M}(l)]^2\,|w(l)-w^{\rm M}(l)|
\le C_1\frac{1}{h_k^{\rm M}}\sum_{l=0}^\infty [p_k^{\rm M}(l)]^2 w^{\rm M}(l)= C_1\,.
\end{aligned}
\end{equation}
Using this in \eqref{h4}, we obtain that
\begin{equation} \label{h7}
\begin{aligned}
\left|\left(\frac{h_k}{h_k^{\rm M}}\right)^{1/2}-\left(\frac{h_k^{\rm M}}{h_k}\right)^{1/2}\right|
\le (C_0C_1)^{1/2}=\frac{e^{-4\ga(m+1)}}{[1-e^{-4\ga(m+1)}]^{1/2}}\,.
\end{aligned}
\end{equation}
This implies that
\begin{equation} \label{h8}
\begin{aligned}
\left|\left(\frac{h_k}{h_k^{\rm M}}\right)-1\right|
\le C_2e^{-4\ga(m+1)}\,,
\end{aligned}
\end{equation}
where $C_2>0$. Since $m\ge n\ep$, estimate \eqref{h2} follows. 

\end{proof}

\section{Proof of Theorem \ref{main}}\label{proofmain}

By \eqref{meix18} and \eqref{h1},
\begin{equation}\label{pr1}
\begin{aligned}
Z_{n-m,n}=Z_{n-m,n}^{\rm M}\prod_{j=0}^{n-m-1} \frac{h_j}{h^{\rm M}_j}=Z_{n-m,n}^{\rm M}\prod_{j=0}^{n-m-1}e^{r_j}
=Z_{n-m,n}^{\rm M}e^{\mathcal O(n e^{-\kappa n})}\,,
\end{aligned}
\end{equation}
hence formula \eqref{mr2} follows from \eqref{meix17}, because $ne^{-\kappa n}=\mathcal O(e^{-\kappa'n})$ for any $\kappa'<\kappa$.

\section{Evaluation of the ratio $h_k/h_k^{\rm M}$ for $0\le m<n$}\label{hkk0}

In this section we prove the following result:

\begin{prop} Fix any $1>\ep>0$. Then  there is a constant $C_\ep>0$ such that
\begin{equation} \label{hp1}
h_k=h_k^{\rm M}e^{r_k},
\end{equation}
where 
\begin{equation} \label{hp2}
|r_k|\le C_\ep  e^{-2\ga m-k^{1-\ep}},
\end{equation}
for all $m$ in the interval $0\le m<n$ and $k\in \Z_+$.
\end{prop}

\begin{proof} From \eqref{h4}-- \eqref{h6}, we obtain that
\begin{equation} \label{hp3}
\begin{aligned}
\left|\left(\frac{h_k}{h_k^{\rm M}}\right)^{1/2}-\left(\frac{h_k^{\rm M}}{h_k}\right)^{1/2}\right|
&\le
\left[\frac{e^{-4\ga(m+1)}}{1-e^{-4\ga(m+1)}}\right]^{1/2}\\
&\times\,\left[\frac{1}{h_k^{\rm M}}
\sum_{l=0}^\infty [p^{\rm M}_k(l)]^2\, w^{\rm M}(l)\,e^{-4\ga (l+m+1)}\right]^{1/2}.
\end{aligned}
\end{equation}
We will estimate the sum in the right hand side by using an explicit formula for the Meixner polynomial $p^{\rm M}_k(l)$.
Let us partition the sum as
\begin{equation} \label{hp4}
\begin{aligned}
\sum_{l=0}^\infty [p^{\rm M}_k(l)]^2\, w^{\rm M}(l)\,e^{-4\ga (l+m+1)}
&=\sum_{l=0}^{L-1} [p^{\rm M}_k(l)]^2\, w^{\rm M}(l)\,e^{-4\ga (l+m+1)}\\
&+\sum_{l=L}^\infty [p^{\rm M}_k(l)]^2\, w^{\rm M}(l)\,e^{-4\ga (l+m+1)},
\end{aligned}
\end{equation}
where
\begin{equation} \label{hp5}
L=\lfloor k^{1-\ep}\rfloor.
\end{equation}
Then
\begin{equation} \label{hp6}
\begin{aligned}
\frac{1}{h_k^{\rm M}}\sum_{l=L}^\infty [p^{\rm M}_k(l)]^2\, w^{\rm M}(l)\,e^{-4\ga (l+m+1)}
&\le e^{-4\ga (L+m+1)}\,\frac{1}{h_k^{\rm M}}\sum_{l=L}^\infty [p^{\rm M}_k(l)]^2\, w^{\rm M}(l)\\
&\le e^{-4\ga (L+m+1)}\le e^{-4\ga k^{1-\ep}-4\ga m},
\end{aligned}
\end{equation}
hence
\begin{equation} \label{hp7}
\begin{aligned}
\frac{e^{-4\ga(m+1)}}{(1-e^{-4\ga(m+1)})\,h_k^{\rm M}}
\sum_{l=L}^\infty [p^{\rm M}_k(l)]^2\, w^{\rm M}(l)\,e^{-4\ga (l+m+1)}
\le Ce^{-4\ga k^{1-\ep}-8\ga m}.
\end{aligned}
\end{equation}
It remains to estimate the term
\begin{equation} \label{hp8}
\de_L=e^{-4\ga(m+1)}\,\frac{1}{h_k^{\rm M}}\sum_{l=0}^{L-1} [p^{\rm M}_k(l)]^2\, w^{\rm M}(l)\,e^{-4\ga (l+m+1)}.
\end{equation}
We may assume that $k\ge 1$, because $\de_L=0$ for $k=0$ (the sum contains no terms for $k=0$).

Let us express $\de_L$ in terms of the Meixner polynomial $M_k(l;m+1,q)$, recalling the notation $q=e^{-2(t-\ga)}$ defined in \eqref{meix11}. By \eqref{meix12},
\begin{equation} \label{hp9}
w^{\rm M}(l)=C_m\,\frac{(m+1)_l q^l}{l!}=q^{m+1}m!\,\frac{(m+1)_l q^{l}}{l!}
=\frac{q^{l+m+1}(l+m)! }{l!}\,.
\end{equation}
Also, by \eqref{meix9} and \eqref{meix12}
\begin{equation} \label{hp10}
p_k^{\rm M}(l)=\frac{(k+m)!}{m!(1-q^{-1})^k} M_k(l;m+1,q),\quad 
h^{\rm M}_k=\frac{k!\,(k+m)!\,q^{k+m+1}}{(1-q)^{2k+m+1}}\,,
\end{equation}
hence
\begin{equation} \label{hp11}
\begin{aligned}
\frac{1}{h_k^{\rm M}}&\sum_{l=0}^{L-1} [p^{\rm M}_k(l)]^2\, w^{\rm M}(l)\,e^{-4\ga (l+m+1)}
=\frac{(1-q)^{2k+m+1}}{k!\,(k+m)!\,q^{k+m+1}}\\
&\times\sum_{l=0}^{L-1}\left[\frac{(k+m)!q^k}{m!(1-q)^k} M_k(l;m+1,q)\right]^2
\, \frac{q^{l+m+1}(l+m)! }{l!}\,e^{-4\ga (l+m+1)}
\\
&=\frac{(k+m)! q^k (1-q)^{m+1}}{k!m!}
\sum_{l=0}^{L-1}   [M_k(l;m+1,q)]^2\, \frac{(l+m)! q^l}{l!m!}\,e^{-4\ga (l+m+1)},
\end{aligned}
\end{equation}
hence
\begin{equation} \label{hp12}
\begin{aligned}
\de_L&=e^{-4\ga(m+1)}\,\frac{1}{h_k^{\rm M}} \sum_{l=0}^{L-1}[p^{\rm M}_k(l)]^2\, w^{\rm M}(l)\,e^{-4\ga (l+m+1)}\\
&= \,\frac{(k+m)! q^k \left[(1-q)e^{-4\ga}\right]^{m+1}}{k!m!}
\sum_{l=0}^{L-1}   [M_k(l;m+1,q)]^2\, \frac{(l+m)! q^l}{l!m!}\,e^{-4\ga (l+m+1)}\,.
\end{aligned}
\end{equation}
To estimate $\frac{(k+m)!}{k!m!}\,,$ we use the inequality
\begin{equation} \label{hp13}
\frac{a^k b^m(k+m)!}{k!m!}\le (a+b)^{k+m},\quad a,b>0.
\end{equation}
Applying this inequality to \eqref{hp13} with 
\begin{equation} \label{hp14}
a=q,\quad b=(1-q)e^{-4\ga},
\end{equation}
we obtain that
\begin{equation} \label{hp15}
\begin{aligned}
\de_L
\le \rho^{k+m}
\sum_{l=0}^{L-1}   [M_k(l;m+1,q)]^2\, \frac{(l+m)! q^l}{l!m!}\,e^{-4\ga (l+m+1)},
\end{aligned}
\end{equation}
where
\begin{equation} \label{hp16}
\begin{aligned}
 \rho=q+(1-q)e^{-4\ga}<1.
\end{aligned}
\end{equation}
Using \eqref{hp14} with $k=l$, $a=e^{2\ga}-1$, and $b=1$, we obtain that
\begin{equation} \label{hp17}
\frac{(l+m)!}{l!m!}\le \frac{e^{2\ga(l+m)}}{(e^{2\ga}-1)^l}\,,
\end{equation}
hence
\begin{equation} \label{hp18}
\begin{aligned}
\de_L \le \rho^{k+m}e^{-2\ga m}
\sum_{l=0}^{L-1}   [M_k(l;m+1,q)]^2\,\al^l\,;\qquad
\al=\frac{q}{e^{2\ga}(e^{2\ga}-1)}=\frac{e^{-2t}}{e^{2\ga}-1}\,.
\end{aligned}
\end{equation}
Let us write $M_k(l;m+1,q)$ starting from the lowest order term:
\begin{equation} \label{hp19}
\begin{aligned}
M_k(l;m+1,q)&=1+\frac{(1-q^{-1})k l}{m+1}+\frac{(1-q^{-1})^2k(k-1)l(l-1)}{2!(m+1)(m+2)}\\
&+\frac{(1-q^{-1})^3 k(k-1)(k-2)l(l-1)(l-2)}{3!(m+1)(m+2)(m+3)}+\cdots.
\end{aligned}
\end{equation}
The latter sum consists of at most $(l+1)$ nonzero terms and for $l\le L-1$ each term is estimated by $(|1-q^{-1}|kL)^{L}$,
hence
\begin{equation} \label{hp20}
\begin{aligned}
M_k(l;m+1,q)\le L(|1-q^{-1}|kL)^{L}.
\end{aligned}
\end{equation}
Using this estimate in \eqref{hp18}, we obtain that
\begin{equation} \label{hp21}
\begin{aligned}
\de_L \le \rho^{k+m}e^{-2\ga m}
L^2(|1-q^{-1}|kL)^{L}\,.
\end{aligned}
\end{equation}
Thus,
\begin{equation} \label{hp22}
\begin{aligned}
\de_L \le \rho^{m}e^{-2\ga m}\exp\left[k\ln \rho+\mathcal O(k^{1-\ep}\ln k)\right]
\le C_\ep e^{-2\ga m-k^{1-\ep}},
\end{aligned}
\end{equation}
for some $C_\ep>0$. From \eqref{hp3}, \eqref{hp7}, and \eqref{hp22}
we obtain that 
\begin{equation} \label{hp23}
\begin{aligned}
\left|\left(\frac{h_k}{h_k^{\rm M}}\right)^{1/2}-\left(\frac{h_k^{\rm M}}{h_k}\right)^{1/2}\right|
\le C_\ep e^{-2\ga m-k^{1-\ep}}
\end{aligned}
\end{equation}
for some $C_\ep>0$. This implies \eqref{hp1}, \eqref{hp2}.
\end{proof}

Substituting \eqref{hp1}, \eqref{hp2} into \eqref{meix18}, we obtain that for any fixed $1>\ep>0$ there is
$C_\ep>0$ such that
\begin{equation} \label{hp24}
\begin{aligned}
&Z_{n-m,m}=C(m)\,Z^{\rm M}_{n-m,m} e^{\xi_{nm}},\quad |\xi_{nm}|
\le C_\ep e^{-2\ga m} \exp\left(\sum_{k=n-m}^\infty e^{-k^{1-\ep}}\right)\,,\\
&C(m)=\prod_{k=0}^\infty \frac{h_k}{h_k^{\rm M}}\,.
\end{aligned}
\end{equation}
This implies that 
for any fixed $1>\ep>0$ there is
$C_\ep>0$ such that 
\begin{equation} \label{hp25}
\begin{aligned}
|\xi_{nm}|
\le C_\ep e^{-2\ga m}  e^{-n^{1-\ep}}\,.
\end{aligned}
\end{equation}
Our next goal will be to calculate the constant factor $C(m)$. From estimate  \eqref{hp2} we have that
as $m\to\infty$,
\begin{equation} \label{hp26}
\begin{aligned}
C(m)=1+\mathcal O(\rho^m),\quad \rho=e^{-2\ga}<1. 
\end{aligned}
\end{equation}

\vskip 3mm

\section{Evaluation of the constant factor $C(m)$}\label{evaluation_C}

In the next two sections we will find the exact value of the constant $C(m)$ in 
formula (\ref{hp24}). This will be done in two steps: first, with the help
of the Toda equation, we will find the form of the dependence of $C(m)$ on $t$,
and second, we will find the large $t$ asymptotics of $C(m)$. By combining 
these two steps, we will obtain the exact value of $C(m)$. In this section
we will carry out the first step of our program.

The weight $w(x)$ in \eqref{meix1} can be written as
  \begin{equation}\label{cm1}
w(x)=e^{-2t(x+m+1)} u(x)\,;\qquad u(x)=2\sinh[2\ga(x+m+1)]\,\frac{(x+m)!}{x!}\,.
    \end{equation}
Since the dependence of $w(x)$ on $t$ is a linear exponent, we have the Toda equation (see e.g. \cite{Bleher-Liechty14}): 
\begin{equation} \label{cm2}
\left(\ln\prod_{k=0}^{n-m-1} h_k\right)''=\frac{4h_{n-m}}{h_{n-m-1}}\,,\qquad \Big(\;\Big)'=\frac{\partial }{\partial t}\,.
\end{equation}
From \eqref{hp1}, \eqref{hp2}, and \eqref{meix12} we obtain that
\begin{equation} \label{cm3}
\frac{h_{n-m}}{h_{n-m-1}}=\frac{h^{\rm M}_{n-m}}{h^{\rm M}_{n-m-1}}\,e^{r_{n-m}-r_{n-m-1}}=
\frac{(n-m)nq}{(1-q)^2}+\mathcal O\left(\rho^m e^{-n^{1-\ep}}\right)\,.
\end{equation}
We have that
\begin{equation} \label{cm4}
\frac{4q}{(1-q)^2}=\frac{4e^{2\ga-2t}}{(1-e^{2\ga-2t})^2}=\left[\frac{(-2)}{1-e^{2\ga-2t}}\right]'
=\left[-\ln(1-e^{2\ga-2t})\right]'',
\end{equation}
hence 
\begin{equation} \label{cm5}
\left(\ln\prod_{k=0}^{n-m-1} h_k\right)''=(n-m)n\left[-\ln(1-e^{2\ga-2t})\right]''+\mathcal O\left(\rho^m e^{-n^{1-\ep}}\right).
\end{equation}
Integrating twice, we obtain that for $t$ in any bounded interval $[t_1,t_2]$ on the line,
\begin{equation} \label{cm6}
\begin{aligned}
\ln\left(\prod_{k=0}^{n-m-1} h_k\right)
&=C_0+C_1t+(n-m)n\left[-\ln(1-e^{2\ga-2t})\right]+\mathcal O\left(\rho^m e^{-n^{1-\ep}}\right)\\
&=C_0+C_1t-(n-m)n\ln(1-q)+\mathcal O\left(\rho^m e^{-n^{1-\ep}}\right).
\end{aligned}
\end{equation}
On the other hand,
from \eqref{as21}, \eqref{meix13}, and \eqref{hp24} we obtain that
\begin{equation} \label{cm7}
\begin{aligned}
\ln\left(\prod_{k=0}^{n-m-1} h_k\right)
=\ln\left(\prod_{k=0}^{n-m-1} h^{\rm M}_k\right)+\ln C(m)+\xi_{nm},
\end{aligned}
\end{equation}
hence
\begin{equation} \label{cm8}
\begin{aligned}
\ln C(m)=C_0+C_1t-(n-m)n\ln(1-q)-\ln\left(\prod_{k=0}^{n-m-1} h^{\rm M}_k\right)+\mathcal O\left(\rho^m e^{-n^{1-\ep}}\right).
\end{aligned}
\end{equation}
By \eqref{meix14},
\begin{equation} \label{cm9}
\begin{aligned}
\ln\left(\prod_{k=0}^{n-m-1} h^{\rm M}_k\right)=C_2+C_3t-(n-m)n\ln(1-q),
\end{aligned}
\end{equation}
where $C_2,C_3$ are independent of $t$, hence
\begin{equation} \label{cm10}
\begin{aligned}
\ln C(m)=C_4+C_5t+\mathcal O\left(\rho^m e^{-n^{1-\ep}}\right),
\end{aligned}
\end{equation}
where $C_4,C_5$ are independent of $t$ (but they may depend on $m,n$). However, $\ln C(m)$ does not depend on $n$ and 
according to the latter equation, as $n\to \infty$ it
is a limit of linear functions of the argument $t$. This implies that $\ln C(m)$ is a linear function of $t$
as well, so that
\begin{equation} \label{cm11}
\begin{aligned}
\ln C(m)=d_0(m)+d_1(m)t\,.
\end{aligned}
\end{equation}
In the next section we will calculate $d_0(m)$ and $d_1(m)$.

\vskip 3mm

\section{Explicit formula for $C(m)$}\label{formula_C}

In this section we find the exact value of $C(m)$, and by doing this we will
finish the proof of Theorem \ref{main2}. 
Consider the following regime:
\begin{equation} \label{fc1}
\ga\; \textrm{is fixed},\;\;m\; \textrm{is fixed},\quad t\to\infty,
\end{equation}
and let us evaluate the asymptotics of $C(m)$ in this regime.
Applying the formula,
\begin{equation} \label{fc2}
\begin{aligned}
\sum_{l=0}^\infty e^{-xl}\prod_{k=1}^m (l+k)
=\frac{m!}{(1-e^{-x })^{m+1}},
\end{aligned}
\end{equation}
to (\ref{as14}), (\ref{as15}), we obtain that
\begin{equation} \label{fc3}
\begin{aligned}
h_0&=\sum_{l=0}^\infty w(l)=\sum_{l=0}^\infty \left[\left(e^{-2(t-\ga) (l+m+1)}-e^{-2(t+\ga) (l+m+1)}\right)\prod_{k=1}^m (l+k)\right]\\
&=m!\left[\frac{e^{-2(t-\ga) (m+1)}}{(1-e^{-2(t-\ga) })^{m+1}}-\frac{e^{-2(t+\ga) (m+1)}}{(1-e^{-2(t+\ga) })^{m+1}}\right]\,.
\end{aligned}
\end{equation}
Similarly,
\begin{equation} \label{fc4}
\begin{aligned}
h^{\rm M}_0&=\sum_{l=0}^\infty w^{\rm M}(l)=\sum_{l=0}^\infty \left[e^{-2(t-\ga) (l+m+1)}\,\prod_{k=1}^m (l+k)\right]
=\frac{m!e^{-2(t-\ga) (m+1)}}{(1-e^{-2(t-\ga) })^{m+1}}\,,
\end{aligned}
\end{equation}
hence as $t\to\infty$,
\begin{equation} \label{fc5}
\begin{aligned}
\frac{h_0}{h^{\rm M}_0}=1-e^{-4\ga(m+1)}\left(\frac{1-e^{-2t+2\ga }}{1-e^{-2t-2\ga }}\right)^{m+1}
=1-e^{-4\ga(m+1)}+\mathcal O(e^{-2t}).
\end{aligned}
\end{equation}
Let us evaluate the quotient $\frac{h_k}{h^{\rm M}_k}$ for $k\ge 1$. We prove the following result:

\begin{prop} Suppose that $\ga$ and $m$ are fixed. Then there are $c>0$ and $t_0>0$ such that
\begin{equation} \label{fc6}
\begin{aligned}
\frac{h_k}{h^{\rm M}_k}=e^{r_k},\quad |r_k|\le  e^{-ct-k^{1/2}},
\end{aligned}
\end{equation}
for all $t\ge t_0$ and $k\ge 1$.
\end{prop} 

\begin{proof} The proof will be based on estimate \eqref{hp3}. We take
\begin{equation} \label{fc7}
L=\lfloor t+k^{2/3} \rfloor.
\end{equation}
Then
\begin{equation} \label{fc8}
\begin{aligned}
\frac{1}{h_k^{\rm M}}\sum_{l=L}^\infty [p^{\rm M}_k(l)]^2\, w^{\rm M}(l)\,e^{-4\ga (l+m+1)}
&\le e^{-4\ga (L+m+1)}\,\frac{1}{h_k^{\rm M}}\sum_{l=L}^\infty [p^{\rm M}_k(l)]^2\, w^{\rm M}(l)\\
&\le e^{-4\ga (t+k^{2/3})}.
\end{aligned}
\end{equation}
It remains to estimate the term
\begin{equation} \label{fc9}
\de_L= e^{-4\ga (m+1)}\frac{1}{h_k^{\rm M}}\sum_{l=0}^{L-1}[p^{\rm M}_k(l)]^2\, w^{\rm M}(l)\,e^{-4\ga (l+m+1)}.
\end{equation}
By \eqref{hp12},
\begin{equation} \label{fc10}
\begin{aligned}
\de_L=\frac{(k+m)! q^k \left[(1-q)e^{-4\ga}\right]^{m+1}}{k!m!}
\sum_{l=0}^{L-1}   [M_k(l;m+1,q)]^2\, \frac{(l+m)! q^l}{l!m!}\,e^{-4\ga (l+m+1)}.
\end{aligned}
\end{equation}
To estimate $\frac{(k+m)!}{k!m!}\,,$ we use inequality \eqref{hp14} with
\begin{equation} \label{fc11}
a=\frac{1-e^{-4\ga}}{2}\,,\quad b=(1-q)e^{-4\ga}.
\end{equation}
This gives 
\begin{equation} \label{fc12}
\begin{aligned}
\de_L
\le \left(\frac{2q}{1-e^{-4\ga}}\right)^k \rho^{k+m}
\sum_{l=0}^{L-1}   [M_k(l;m+1,q)]^2\, \frac{(l+m)! q^l}{l!m!}\,e^{-4\ga (l+m+1)},
\end{aligned}
\end{equation}
where
\begin{equation} \label{fc13}
\begin{aligned}
\rho= \frac{1+e^{-4\ga}}{2}<1.
\end{aligned}
\end{equation}
The key point here that we still have the factor $q^k$ in \eqref{fc12} on the right, where $q=e^{-2t+2\ga}$ is exponentially small
as $t\to\infty$. Similar to \eqref{hp23}, we obtain that
\begin{equation} \label{fc14}
\begin{aligned}
\left|\left(\frac{h_k}{h_k^{\rm M}}\right)^{1/2}-\left(\frac{h_k^{\rm M}}{h_k}\right)^{1/2}\right|
\le C\left(\frac{2q}{1-e^{-4\ga}}\right)^k \rho^{m}\exp\left(-k^{2/3}\right)
\end{aligned}
\end{equation}
for some $C>0$. Together with \eqref{fc8} this proves \eqref{fc6}.
\end{proof}

Using formulae \eqref{meix18}, \eqref{fc5}, and \eqref{fc6}, we can calculate $C(m)$. Namely, from these formulae we obtain that 
\begin{equation}\label{fc15}
\begin{aligned}
\frac{Z_{n-m,n}}{Z_{n-m,n}^{\rm M}}
=\prod_{k=0}^{n-m-1} \frac{h_k}{h^{\rm M}_k}=\left[1-e^{-4\ga(m+1)}+\mathcal O(e^{-2t})\right]\prod_{k=1}^{n-m-1} e^{r_k},\quad
|r_k|\le  e^{-ct-k^{1/2}}\,,
\end{aligned}
\end{equation}
hence by \eqref{hp24}, as $t\to\infty$,
\begin{equation}\label{fc16}
\begin{aligned}
C(m)=\lim_{n\to\infty}\frac{Z_{n-m,n}}{Z_{n-m,n}^{\rm M}}
=1-e^{-4\ga(m+1)}+\mathcal O(e^{-2t})\,,
\end{aligned}
\end{equation}
so that
\begin{equation}\label{fc17}
\begin{aligned}
\ln C(m)
=\ln\left[1-e^{-4\ga(m+1)}\right]+\mathcal O(e^{-2t})\,.
\end{aligned}
\end{equation}
Comparing this with \eqref{cm11}, we conclude that $d_0(m)=\ln\left[1-e^{-4\ga(m+1)}\right]$ and $d_1(m)=0$, hence
\begin{equation}\label{fc18}
\begin{aligned}
C(m)=1-e^{-4\ga(m+1)}\,.
\end{aligned}
\end{equation}

\section{Asymptotics of orthogonal polynomials: a phase transition}\label{phase}

The interpolation problem discussed in Section \ref{RHA} can be used to obtain an asymptotic formula 
for the orthogonal polynomials $p_k(z)$ 
with respect to the weight $w(x)=w(x;m)$ defined in \eqref{meix1}. We consider here a scaling regime, when $m,k\to\infty$
in such a way that  $m=k\xi$ where $0\le \xi\le A$ for some $A>0$. To describe the corresponding equilibrium measure, 
introduce the potential function
\begin{equation}\label{pt1}
V(x)=2(t-\ga)x+x\ln x-x\ln(x+\xi)-\xi\ln(x+\xi)+\xi,
    \end{equation}
and the energy functional
   \begin{equation}\label{pt2}
I_{V}(\nu)=-\iint_{x\not=y} \log |x-y| d\nu(x)d\nu(y)+\int V(x) d\nu(x).
    \end{equation}
The equilibrium measure $\nu_{\rm eq}$ minimizes $I_{V}(\nu)$ over the space of probability measures $\nu$ on the line
with the constraint
  \begin{equation}\label{pt3}
\nu E\le m E,
    \end{equation}
for any measurable set $E$, where $m E$ is the Lebesgue measure. The equilibrium measure is an essential part of the steepest descent analysis of the interpolation problem, and in particular gives the limiting density of zeroes of the polynomials $p_k$ after a rescaling as $k\to \infty$.

An analysis of the minimization problem (see \cite[Section 6]{Johansson00}) reveals a phase transition at $\xi=\xi_c$, where
  \begin{equation}\label{pt4}
\xi_c=e^{2t-2\ga}-1.
    \end{equation}
Namely, for $0\le \xi <\xi_c$ there are numbers $0<a<b$ such that the equilibrium measure $\nu_{\rm eq}$ is saturated on the interval $[0,a]$ so that
  \begin{equation}\label{pt5}
\frac{d\nu_{\rm eq}(x)}{dx}=1,\quad 0\le x\le a,
    \end{equation}
and $\nu_{\rm eq}$ has a band on the interval $(a,b)$, so that 
  \begin{equation}\label{pt6}
0<\frac{d\nu_{\rm eq}(x)}{dx}<1,\quad a< x< b.
    \end{equation}
Finally, the interval $[0,\infty)$ is a void one, so that 
  \begin{equation}\label{pt7}
\frac{d\nu_{\rm eq}(x)}{dx}=0,\quad  x\ge  b.
    \end{equation}
For $\xi>\xi_c$, there is no saturated interval, and the equilibrium measure is supported by a band $(a,b)$, where $0<a<b$. 

It is interesting to notice that the phase transition in the equilibrium problem has no effect on the asymptotic behavior
of the partition function $Z_{n-m,m}$ in Theorems \ref{main}, \ref{main2}.

\begin{appendix}

\section{Proof of Proposition \ref{conlaws}}\label{app0}

\begin{figure}
\begin{center}\scalebox{0.45}{\includegraphics{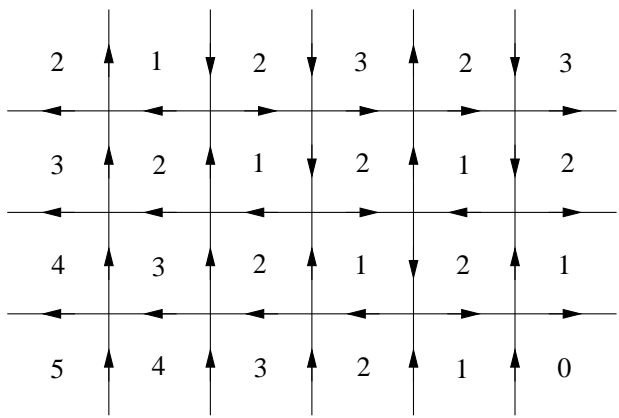}}\end{center}
\caption{The height function.}
\label{height_function}
\end{figure}

To prove the last equation in \eqref{int3}, fix a configuration $\sg$ and consider the corresponding
height function $h(v)$ defined on the faces of the lattice
(or on the vertices of the dual lattice $V'$)
by the condition that for any two neighboring faces $v,w$, 
\begin{equation}\label{hf1}
h(w)-h(v)=(-1)^s,
\end{equation}
where $s=0$ if the arrow $\sg_e$ on the edge $e\in E$, crossing the segment $[v,w]$, is oriented in such a way that
it points from left to right with respect to the vector $\vec{vw}$ , and $s=1$ if $\sg_e$ is oriented 
from right to left with respect to $\vec{vw}$.  The ice-rule ensures that the height function $h=h_\sg$ exists for any configuration $\sg$. 
 An example of a configuration and its corresponding height function is given in Figure \ref{height_function}.
The height function is defined up to an additive constant, and we fixed it by assigning 0 to the face in the right lower corner.
Observe that due to the partial domain wall boundary conditions, 
on the boundary the height function is linear on the left and right sides, and on the lower boundary. 
Introduce the coordinates on the dual lattice such that the origin is at the right lower corner, and the $x$-axis going
left and the $y$-axis going up.
Then on the left and right sides, and on the lower boundary, 
\begin{equation}\label{hf2}
\begin{aligned}
&h(0,k)=k,\quad 0\le k\le n-m,\\
&h(j,0)=j,\quad 0\le j\le n,\\
&h(n,k)=n-k,\quad 0\le k\le n-m.
\end{aligned}
\end{equation}

The height function can be used to calculate the differences $N_2(\sg)-N_1(\sg)$ and $N_4(\sg)-N_3(\sg)$. 
Consider any line $L$ on the dual lattice parallel to the diagonal $y=x$. 
Then along this line the height function jumps by 2  on any vertex configuration
of type 1 and by $(-2)$ on any vertex configuration of type 2. The height function does not change on any vertex configuration of types
3, 4, 5, 6 (See Figure \ref{height_function_arrows}).

\begin{figure}
\begin{center}\scalebox{0.45}{\includegraphics{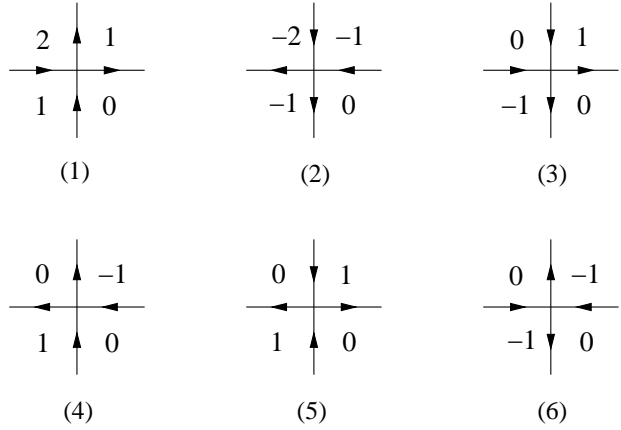}}\end{center}
\caption{The height function on vertex arrow configurations.}
\label{height_function_arrows}
\end{figure}

Let $v_1,\ldots,v_k$ be the vertices of the dual lattice $V'$ along the line $L$. Then 
\begin{equation}\label{hf3}
h(v_k)-h(v_1)=2N_1(\sg,L)-2N_2(\sg,L),
\end{equation}
where $N_i(\sg;L)$ is the number of vertex states of type $i$ in $\sg$
on the line $L$. By summing up over all possible lines $L$, we obtain that
\begin{equation}\label{hf4}
H-S=2N_1(\sg)-2N_2(\sg),
\end{equation}
where 
$H$ is the sum of the heights $h(v)$ along the top row,
\begin{equation}\label{hf5}
H=h(1,n-m)+h(2,n-m)+\cdots+h(n-1,n-m).
\end{equation}
and
\begin{equation}\label{hf6}
\begin{aligned}
S&=[1+\cdots+(n-m-1)]+[1+\cdots+(m-1)]\\
&=\frac{(n-m-1)(n-m)}{2}+\frac{(m-1)m}{2}\,.
\end{aligned}
\end{equation}
Similarly, summing up along the lines parallel to the diagonal $y=-x$, we obtain that
\begin{equation}\label{hf7}
H-T=2N_3(\sg)-2N_4(\sg),
\end{equation}
where 
\begin{equation}\label{hf8}
\begin{aligned}
T&=[(m+1)+\cdots+n]+[(n-1)+\cdots+(n-m+1)]\\
&=\frac{(n-m)(n+m+1)}{2}+\frac{(m-1)(2n-m)}{2}\,.
\end{aligned}
\end{equation}
Since
\begin{equation}\label{hf9}
T-S=2m(n-m),
\end{equation}
we obtain from \eqref{hf4} and \eqref{hf7} that
\begin{equation}\label{hf10}
[N_1(\sg)-N_2(\sg)]-[N_3(\sg)-N_4(\sg)]=\frac{T-S}{2}=m(n-m).
\end{equation}
This proves the last equation in \eqref{int3}.

\newpage

 \section{Proof of proposition \ref{part_function_det_formula}}\label{appendix1}
We begin with a partially inhomogeneous six-vertex model with DWBC.  That is, consider the $n\times n$ square lattice with parameters $(\la_1, \dots, \la_n)$ assigned to horizontal lines from top to bottom, see Fig. \ref{inhomogeneous}.  
\begin{figure}
\begin{center}\scalebox{0.3}{\includegraphics{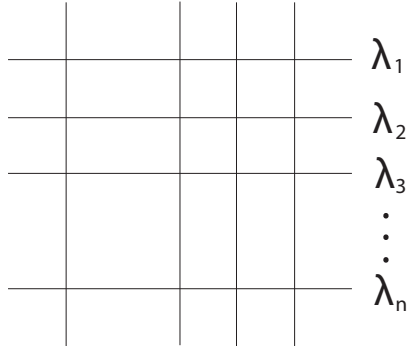}}\end{center}
\caption{The $n\times n$ square lattice with spectral parameters $(\la_1, \dots, \la_n)$. }
\label{inhomogeneous}
\end{figure}
We label the six vertex types as in Fig. \ref{arrows}, and use different weights in each row:
\begin{equation}\label{ap1}
w_{j} = \left\{
\begin{aligned}
&a_-(\la_j)&:= &e^{-\ga} a(\la_j) \qquad &\textrm{if vertex in row } j \textrm{ is of type 1} \\
&a_+(\la_j)&:=&e^{\ga} a(\la_j) \qquad &\textrm{if vertex in row } j \textrm{ is of type 2} \\
&b_-(\la_j)&:= &e^{-\ga} b(\la_j) \qquad &\textrm{if vertex in row } j \textrm{ is of type 3} \\
&b_+(\la_j)&:= &e^{\ga} b(\la_j) \qquad &\textrm{if vertex in row } j \textrm{ is of type 4}  \\
&c(\la_j)&:=&\sinh(2\ga) \qquad &\textrm{if vertex in row } j \textrm{ is of type 5 or 6}\,,
\end{aligned}\right.
\end{equation}
where \begin{equation}\label{ap2}
a(\la) = \sinh(\la-\ga)\,, \qquad b(\la) = \sinh(\la+\ga)\,, \qquad c(\la)\equiv c=\sinh(2\ga).
\end{equation}

Introduce the notations
\begin{equation}\label{ap3}
\begin{aligned}
\varphi(\la):&=a(\la) b(\la)= \sinh(\la-\ga)\sinh(\la+\ga)\,,\\  \phi(\la):&=\frac{\sinh(2\ga)}{\sinh(\la-\ga)\sinh(\la+\ga)}. \\
\end{aligned}
\end{equation}
The Izergin-Korepin formula for the partially inhomogeneous partition function is \cite{Izergin87}, \cite{Izergin-Coker-Korepin92}
\begin{equation}\label{ap4}
Z_n^{\textrm{inh}}= \frac{(-1)^{n(n-1)/2}\prod_{j=1}^{n} \varphi(\la_j)^n}{\prod_{j=0}^{n-1} j! \prod_{1\le j <k \le n} \sinh(\la_j-\la_k)} \det\big(\phi^{(k-1)}(\la_j)\big)_{j,k=1}^n,
\end{equation}
where $\phi^{(k)}$ is the $k$th derivative of $\phi$. Observe that the factor $(-1)^{n(n-1)/2}$ comes from our ordering
of $\sinh(\la_j-\la_k)$ in the denominator.

Now introduce the following notations.   Let $Z_{n-m,n}^{\textrm{inh}}$ be the partition function for the six-vertex model on the $(n-m) \times n$ lattice with the parameters $(\la_{m+1}, \dots, \la_n)$, with arrows pointing out on the left and right boundaries, in on the bottom boundary, and the top boundary free.  On the top boundary, there are exactly $m$ arrows pointing up, and $n-m$ arrows pointing down.  For an $m$-tuple of integers $1 \le k_1 < k_2 < \dots <k_m\le n$, consider the partially inhomogeneous six-vertex model on the $(n-m) \times n$ lattice with the following fixed boundary conditions: arrows on left and right boundaries point out, arrows on bottom boundary point in, and the up-pointing arrows on the top boundary are placed $k_1$th, $k_2$th, ..., and $k_m$th location from the right.  We denote the partition function of this model with parameters  $(\la_{m+1}, \dots, \la_n)$ by $Z_{n-m,n}^{\textrm{inh}(k_1, k_2, \dots, k_m)}$   Clearly then we have
\begin{equation}\label{ap5}
Z_{n-m,n}^{\textrm{inh}}= \sum_{1\le k_1<k_2<\cdots<k_m \le n} Z_{n-m,n}^{\textrm{inh}(k_1, k_2, \dots, k_m)}\,.
\end{equation}
For what follows, we set $Z_{n,n}^{\textrm{inh}}=Z_{n}^{\textrm{inh}}$.

Introduce the notation
\begin{equation}\label{ap6}
f_r(\ga)=e^{2\ga r}+e^{2\ga (r-2)}+e^{2\ga (r-4)}+\cdots +e^{-2\ga r}
=\frac{e^{2\ga (r+1)}-e^{-2\ga (r+1)}}{e^{2\ga }-e^{-2\ga }}\,.
\end{equation}
The formula for $Z_{n-m,n}^{\textrm{inh}}$ follows from the following inductive lemma.

\begin{lem}\label{inductive_lemma}
The partition function $Z_{n-m-1,n}^{\textrm{inh}}$ is obtained from $Z_{n-m,n}^{\textrm{inh}}$ via the limit,
\begin{equation}\label{ap7}
 Z_{n-m-1,n}^{\textrm{inh}}=\frac{2^{n-1}}{c f_m(\ga)}\lim_{\la_{m+1}\to \infty} e^{-(n-1)\la_{m+1}}Z_{n-m,n}^{\textrm{inh}}\,.
\end{equation}
\end{lem}


\begin{proof}
For a configuration on the $(n-m)\times n$ lattice, let us consider the weight of the first row when there is exactly one $c$-type vertex in that row.  This can happen when there is an up-pointing arrow in the second row of arrows directly below each up-pointing arrow in the first row.  The remaining up-pointing arrow in the second row of arrows may be placed anywhere else, and gives the $c$-type vertex in the first row of vertices.  Counting the weight of the first row of vertices we find
\begin{equation}\label{ap8}
\begin{aligned}
Z_{n-m,n}^{\textrm{inh}(k_1, k_2,\dots,k_m)}=&c\,\bigg[\sum_{1\le l < k_1} b_-(\la_{m+1})^{l-1} a_+(\la_{m+1})^{n-l-m} b_+(\la_{m+1})^{m}
Z_{n-m-1,n}^{\textrm{inh}(l, k_1, k_2,\dots,k_m)} \\
&+  \sum_{k_1 < l < k_2} b_-(\la_{m+1})^{l-2} a_-(\la_{m+1}) a_+(\la_{m+1})^{n-l-m+1} b_+(\la_{m+1})^{m-1}Z_{n-m-1,n}^{\textrm{inh}(k_1, l, k_2,\dots,k_m)} \\
&+  \sum_{k_2 < l < k_3} b_-(\la_{m+1})^{l-3} a_-(\la_{m+1})^2 a_+(\la_{m+1})^{n-l-m+2} b_+(\la_{m+1})^{m-2}Z_{n-m-1,n}^{\textrm{inh}(k_1, k_2, l ,k_3\dots,k_m)} \\
&\qquad \vdots \\
&+  \sum_{k_{m-1} < l < k_m} b_-(\la_{m+1})^{l-m} a_-(\la_{m+1})^{m-1} a_+(\la_{m+1})^{n-l-1} b_+(\la_{m+1})Z_{n-m-1,n}^{\textrm{inh}(k_1,\dots,k_{m-1}, l, k_m)} \\
&+  \sum_{k_{m} < l \le n} b_-(\la_{m+1})^{l-m-1} a_-(\la_{m+1})^{m} a_+(\la_{m+1})^{n-l}Z_{n-m-1,n}^{\textrm{inh}(k_1, k_2,\dots,k_m, l)}\bigg] \\
&\quad + \textrm{weights of configurations with more than one } c\textrm{-type vertex in first row}\,.
\end{aligned}
\end{equation} 

Now consider the limit as $\la_{m+1} \to +\infty$.  In this limit we have
\begin{equation}\label{ap9}
\begin{aligned}
&a_+(\la_{m+1})=\frac{e^{\la_{m+1}}}{2}\left(1+\bigO(e^{-2\la_{m+1}})\right)\,, 
\quad a_-(\la_{m+1})=\frac{e^{\la_{m+1}}e^{-2\ga}}{2}\left(1+\bigO(e^{-2\la_{m+1}})\right)\,,\\
&b_+(\la_{m+1})=\frac{e^{\la_{m+1}}e^{2\ga}}{2}\left(1+\bigO\left(e^{-2\la_{m+1}}\right)\right)\,,
\quad b_-(\la_{m+1})=\frac{e^{\la_{m+1}}}{2}\left(1+\bigO\left(e^{-2\la_{m+1}}\right)\right)\,, 
\end{aligned}
\end{equation}
and configurations with more than one $c$-type vertex in the first row are $\bigO(e^{(n-2)\la_{m+1}})$.  We therefore find
\begin{equation}\label{ap10}
\begin{aligned}
Z_{n-m,n}^{\textrm{inh}(k_1, k_2,\dots,k_m)}&=\frac{e^{(n-1)\la_{m+1}}\,c}{2^{n-1}}
\bigg[e^{2m\ga} \sum_{1\le l < k_1} Z_{n-m-1,n}^{\textrm{inh}(l, k_1, k_2,\dots,k_m)}\\
&+e^{2(m-2)\ga}\sum_{k_1 < l < k_2}Z_{n-m-1,n}^{\textrm{inh}(k_1, l, k_2,\dots,k_m)} 
+\quad  \dots \\
& +e^{-2(m-2)\ga}\sum_{k_{m-1} < l < k_m} Z_{n-m-1,n}^{\textrm{inh}( k_1,k_2,\dots,k_{m-1}, l, k_m)} \\
&+ e^{-2m\ga}\sum_{k_{m} < l \le n} Z_{n-m-1,n}^{\textrm{inh}(k_1,k_2,\dots,k_{m-1}, k_m, l)}\bigg](1+\bigO(e^{-\la_{m+1}}))\,.
\end{aligned}
\end{equation}
Taking the sum over all ordered $m$-tuples $1\le k_1<k_2< \dots <k_m\le n$, we find
\begin{equation}\label{ap11}
\begin{aligned}
&\sum_{1\le k_1<k_2<\cdots<k_m \le n}Z_{n-m,n}^{\textrm{inh}(k_1, k_2,\dots,k_m)}
=\frac{e^{(n-1)\la_{m+1}}\,c}{2^{n-1}}\\
&\times\bigg[e^{2m\ga} \sum_{1\le k_1<k_2<\cdots<k_m \le n}\;\sum_{1\le l < k_1} Z_{n-m-1,n}^{\textrm{inh}(l, k_1, k_2,\dots,k_m)}\\
&+e^{2(m-2)\ga}\sum_{1\le k_1<k_2<\cdots<k_m \le n}\;\sum_{k_1 < l < k_2}Z_{n-m-1,n}^{\textrm{inh}(k_1, l, k_2,\dots,k_m)} 
+\quad  \dots \\
&+e^{-2(m-2)\ga}\sum_{1\le k_1<k_2<\cdots<k_m \le n}\;\sum_{k_{m-1} < l < k_m} Z_{n-m-1,n}^{\textrm{inh}( k_1,k_2,\dots,k_{m-1}, l, k_m)} \\
&+ e^{-2m\ga}\sum_{1\le k_1<k_2<\cdots<k_m \le n}\;\sum_{k_{m} < l \le n} Z_{n-m-1,n}^{\textrm{inh}(k_1,k_2,\dots,k_{m-1}, k_m, l)}\bigg](1+\bigO(e^{-\la_{m+1}}))\,.
\end{aligned}
\end{equation}
By \eqref{ap5}, the left-hand side of the latter equation is equal to $Z_{n-m,n}^{\textrm{inh}}$. Also, by \eqref{ap5},
\begin{equation}\label{ap12}
\begin{aligned}
&\sum_{1\le k_1<k_2<\cdots<k_m \le n}\;\sum_{1\le l < k_1} Z_{n-m-1,n}^{\textrm{inh}(l, k_1, k_2,\dots,k_m)}\\
&=\sum_{1\le k_1<k_2<\cdots<k_m \le n}\;\sum_{k_1 < l < k_2}Z_{n-m-1,n}^{\textrm{inh}(k_1, l, k_2,\dots,k_m)}=\ldots\\
&=\sum_{1\le k_1<k_2<\cdots<k_m \le n}\;\sum_{k_{m-1} < l < k_m} Z_{n-m-1,n}^{\textrm{inh}( k_1,k_2,\dots,k_{m-1}, l, k_m)} \\
&=\sum_{1\le k_1<k_2<\cdots<k_m \le n}\;\sum_{k_{m} < l \le n} Z_{n-m-1,n}^{\textrm{inh}(k_1,k_2,\dots,k_{m-1}, k_m, l)}
= Z_{n-m-1,n}^{\textrm{inh}},
\end{aligned}
\end{equation}
hence from \eqref{ap11} we obtain that
\begin{equation}\label{ap13}
Z_{n-m,n}^{\textrm{inh}}= \frac{e^{(n-1)\la_{m+1}}\,c}{2^{n-1}} Z_{n-m-1,n}^{\textrm{inh}} f_m(\ga)(1+\bigO(e^{-\la_{m+1}}))\,,
\end{equation}
Taking the limit as $\la_{m+1} \to \infty$, we obtain \eqref{ap7}, and lemma \ref{inductive_lemma} is proved.
\end{proof}
\begin{rem}
Notice that the coefficient of each of the fixed-boundary partition functions on the right-hand side of \eqref{ap10} does not depend on $l$, even though the analogous coefficients in \eqref{ap8} (before taking $\la_{m+1} \to\infty$) do depend on $l$. This is a consequence of the particular asymptotics \eqref{ap9}, which in turn follow from the particular choice of weights \eqref{ap1}. If we let $a_\pm(\la_j)=a(\la_j)e^{\pm \eta}$ and $b_\pm(\la_j)=b(\la_j)e^{\pm \eta}$ for $\eta \ne \ga$ (see \eqref{int12}), then the $l$-dependence of these coefficients persists in \eqref{ap10}. In this case the multi-sums on the right-hand side of \eqref{ap11} do not yield the pDWBC partition function.
\end{rem}

We can apply this lemma inductively, starting from 
\begin{equation}\label{ap14}
Z_{n,n}^{\textrm{inh}} \equiv Z_{n}^{\textrm{inh}}
=\frac{(-1)^{n(n-1)/2}\prod_{j=1}^{n} \varphi(\la_j)^n}{\prod_{j=0}^{n-1} j! \prod_{1\le j <k \le n} \sinh(\la_j-\la_k)} \det\big(\phi^{(k-1)}(\la_j)\big)_{j,k=1}^n\,.
\end{equation}
Namely, we have the following proposition:

\begin{prop}\label{inhomo_pf}
The partition function $Z_{n-m,n}^{\textrm{inh}}$ is given by
\begin{equation}\label{ap15}
\begin{aligned}
Z_{n-m,n}^{\textrm{inh}} = &\frac{(-1)^{n(n-1)/2}\prod_{j=m+1}^n \big[e^{m\la_j}\varphi(\la_j)^n \big] }
{2^{m(m-1)/2} \prod_{j=0}^{n-1}j! \prod_{m+1 \le j<k\le n} \sinh(\la_j-\la_k)} \\
&\times\det \begin{pmatrix} 1 & (-2) & (-2)^2 & \dots & (-2)^{n-1} \\ \vdots & \vdots &\vdots & \ddots & \vdots \\ 1 & (-2m) & (-2m)^2 & \dots & (-2m)^{n-1} \\ \phi(\la_{m+1}) & \phi'(\la_{m+1}) & \phi''(\la_{m+1}) & \dots & \phi^{(n-1)}(\la_{m+1}) \\
\vdots & \vdots &\vdots & \ddots & \vdots \\  \phi(\la_n) & \phi'(\la_n) & \phi''(\la_n) & \dots & \phi^{(n-1)}(\la_n) \end{pmatrix}.
\end{aligned}
\end{equation}
\end{prop}

\begin{proof} From \eqref{ap14},
\begin{equation}\label{ap16}
\begin{aligned}
\lim_{\la_1 \to\infty} e^{-(n-1)\la_1} Z_{n,n}^{\textrm{inh}} 
&=\lim_{\la_1 \to\infty} 
\frac{(-1)^{n(n-1)/2}e^{-(n-1)\la_1}\prod_{j=1}^{n} \varphi(\la_j)^n}{\prod_{j=0}^{n-1} j! \prod_{1\le j <k \le n} \sinh(\la_j-\la_k)}\\
&\times  \det \begin{pmatrix}\phi(\la_1) & \phi'(\la_1) & \phi''(\la_1) & \dots & \phi^{(n-1)}(\la_1) \\ 
\phi(\la_2) & \phi'(\la_2) & \phi''(\la_2) & \dots & \phi^{(n-1)}(\la_2) \\
\vdots & \vdots &\vdots & \ddots & \vdots \\  \phi(\la_n) & \phi'(\la_n) & \phi''(\la_n) & \dots & \phi^{(n-1)}(\la_n) \end{pmatrix}.
\end{aligned}
\end{equation}
Notice that as $\la_j \to \infty$,
\begin{equation}\label{ap17}
\begin{aligned}
&\varphi(\la_j) = \sinh(\la_j-\ga)\sinh(\la_j+\ga)=\frac{e^{2\la_j}}{4}\left(1+\bigO(e^{-2\la_j})\right)\,, \\
&\sinh(\la_j -\la_k)= \frac{e^{\la_j-\la_k}}{2}\left(1+\bigO(e^{-2\la_j})\right). 
\end{aligned}
\end{equation}
Consider now 
\begin{equation}\label{ap18}
\begin{aligned}
\phi(\la_j)&=\frac{\sinh(2\ga)}{\sinh(\la_j-\ga)\sinh(\la_j+\ga)}=\frac{4\sinh(2\ga)}
{\big(e^{\la_j-\ga}-e^{-\la_j+\ga}\big)\big(e^{\la_j+\ga}-e^{-\la_j-\ga}\big)}\\
&=\frac{4\sinh(2\ga)}
{e^{2\la_j}\big(e^{-2\ga}-e^{-2\la_j}\big)\big(e^{2\ga}-e^{-2\la_j}\big)}\\
&=\frac{4\sinh(2\ga)}
{e^{2\la_j}\big(e^{2\ga}-e^{-2\ga}\big)}
\left(\frac{1}{e^{-2\ga}-e^{-2\la_j}}-\frac{1}{e^{2\ga}-e^{-2\la_j}}\right)\\
&=\frac{4\sinh(2\ga)}
{\big(e^{2\ga}-e^{-2\ga}\big)}\sum_{q=1}^\infty \left( e^{2q\ga }-e^{-2q\ga }\right)e^{-2q\la_j}
=4\sinh(2\ga)
\sum_{r=0}^\infty f_r(\ga) e^{-2(r+1)\la_j},
\end{aligned}
\end{equation}
where $f_r(\ga)$ is defined in \eqref{ap6} (we set $q=r+1$ in the last line). Differentiating $k$ times, we obtain that
\begin{equation}\label{ap19}
\begin{aligned}
\phi^{(k)}(\la_j)&=4\sinh(2\ga)
\sum_{r=0}^\infty f_r(\ga)[-2(r+1)]^k e^{-2(r+1)\la_j},
\end{aligned}
\end{equation}
Keeping the term $r=0$ only and taking $j=1$, we have that
\begin{equation}\label{ap20}
\begin{aligned}
\phi^{(k)}(\la_1)&=4\sinh(2\ga)f_0(\ga)
(-2)^k e^{-2\la_1}+\bigO(e^{-4\la_1}).
\end{aligned}
\end{equation}
Substituting the latter formula into \eqref{ap16}, we obtain that
\begin{equation}\label{ap21}
\begin{aligned}
\lim_{\la_1 \to\infty} e^{-(n-1)\la_1} Z_n^{\textrm{inh}} 
&=\lim_{\la_1 \to\infty} 
\frac{(-1)^{n(n-1)/2}4\sinh(2\ga)f_0(\ga)\,e^{-(n+1)\la_1}\prod_{j=1}^{n} \varphi(\la_j)^n}
{\prod_{j=0}^{n-1} j! \prod_{1\le j <k \le n} \sinh(\la_j-\la_k)}\\
&\times  \det \begin{pmatrix}1 & (-2) & (-2)^2 & \dots & (-2)^{n-1} \\ 
\phi(\la_2) & \phi'(\la_2) & \phi''(\la_2) & \dots & \phi^{(n-1)}(\la_2) \\
\vdots & \vdots &\vdots & \ddots & \vdots \\  \phi(\la_n) & \phi'(\la_n) & \phi''(\la_n) & \dots & \phi^{(n-1)}(\la_n) \end{pmatrix}.
\end{aligned}
\end{equation}
Now, from \eqref{ap17} we find that
\begin{equation}\label{ap22}
\begin{aligned}
\lim_{\la_1 \to\infty} 
\frac{e^{-(n+1)\la_1} \varphi(\la_1)^n}{ \prod_{k=2}^n \sinh(\la_1-\la_k)}
&=\lim_{\la_1 \to\infty} 
\frac{e^{-(n+1)\la_1} \left(\frac{e^{2\la_1}}{4}\right)^n}
{\prod_{k=2}^n\frac{e^{\la_1-\la_k}}{2}}
=\frac{1}{2^{n+1}}\,\prod_{k=2}^{n} e^{\la_k},
\end{aligned}
\end{equation}
hence
\begin{equation}\label{ap23}
\begin{aligned}
\lim_{\la_1 \to\infty} e^{-(n-1)\la_1} Z_n^{\textrm{inh}} 
&= \frac{(-1)^{n(n-1)/2}\sinh(2\ga)f_0(\ga)\prod_{j=2}^{n} \left[e^{\la_j}\varphi(\la_j)^n\right]}{2^{n-1}\prod_{j=0}^{n-1} j! \prod_{2\le j <k \le n} \sinh(\la_j-\la_k)}\\
&\times  \det \begin{pmatrix}1 & (-2) & (-2)^2 & \dots & (-2)^{n-1} \\ 
\phi(\la_2) & \phi'(\la_2) & \phi''(\la_2) & \dots & \phi^{(n-1)}(\la_2) \\
\vdots & \vdots &\vdots & \ddots & \vdots \\  \phi(\la_n) & \phi'(\la_n) & \phi''(\la_n) & \dots & \phi^{(n-1)}(\la_n) \end{pmatrix}.
\end{aligned}
\end{equation}
Thus, by \eqref{ap7} [remind that $c=\sinh(2\ga)$],
\begin{equation}\label{ap24}
\begin{aligned}
Z_{n-1,n}^{\textrm{inh}}&=\frac{2^{n-1}}{c f_0(\ga)}\lim_{\la_{1}\to \infty} e^{-(n-1)\la_{1}}Z_{n}^{\textrm{inh}} 
= \frac{(-1)^{n(n-1)/2}\prod_{j=2}^{n} \left[e^{\la_j}\varphi(\la_j)^n\right]}{\prod_{j=0}^{n-1} j! \prod_{2\le j <k \le n} \sinh(\la_j-\la_k)}\\
&\times  \det \begin{pmatrix}1 & (-2) & (-2)^2 & \dots & (-2)^{n-1} \\ 
\phi(\la_2) & \phi'(\la_2) & \phi''(\la_2) & \dots & \phi^{(n-1)}(\la_2) \\
\vdots & \vdots &\vdots & \ddots & \vdots \\  \phi(\la_n) & \phi'(\la_n) & \phi''(\la_n) & \dots & \phi^{(n-1)}(\la_n) \end{pmatrix}.
\end{aligned}
\end{equation}

We now consider the limit of $e^{-(n-1)\la_2} Z_{n-1,n}^{\textrm{inh}}$ as $\la_2 \to \infty$.  
To that end, we keep in \eqref{ap19} terms with $r=0$ and $r=1$:
\begin{equation}\label{ap25}
\begin{aligned}
\phi^{(k)}(\la_2)&=4\sinh(2\ga)\left[
(-2)^k e^{-2\la_2}+f_1(\ga) (-4)^k e^{-4\la_2}\right]+\bigO(e^{-6\la_2}).
\end{aligned}
\end{equation}
Substituting this into the second row of the determinant in \eqref{ap24} and taking a linear combination with the first row, we obtain that
\begin{equation}\label{ap26}
\begin{aligned}
\lim_{\la_2 \to\infty}e^{-(n-1)\la_2}& Z_{n-1,n}^{\textrm{inh}} 
= \lim_{\la_2 \to\infty}\frac{(-1)^{n(n-1)/2} 4\sinh(2\ga)  f_1(\ga)e^{-(n+3)\la_2}
\prod_{j=2}^n \left[e^{\la_j}\varphi(\la_j)^n\right]  }{ \prod_{j=0}^{n-1}j!
\prod_{2 \le j<k\le n} \sinh(\la_j-\la_k)} \\
&\qquad \times \det \begin{pmatrix} 1 & (-2) & (-2)^2 & \dots & (-2)^{n-1} \\ 1 & (-4) & (-4)^2 & \dots & (-4)^{n-1} \\ \phi(\la_3) & \phi'(\la_3) & \phi''(\la_3) & \dots & \phi^{(n-1)}(\la_3) \\
\vdots & \vdots &\vdots & \ddots & \vdots \\  \phi(\la_n) & \phi'(\la_n) & \phi''(\la_n) & \dots & \phi^{(n-1)}(\la_n) \end{pmatrix}.
\end{aligned}
\end{equation}
Now, 
\begin{equation}\label{ap27}
\begin{aligned}
\lim_{\la_2 \to\infty} 
\frac{e^{-(n+3)\la_2} \left[ e^{\la_2}\varphi(\la_2)^n\right]}{ \prod_{k=3}^n \sinh(\la_2-\la_k)}
&=\lim_{\la_2 \to\infty} 
\frac{e^{-(n+2)\la_2} \left(\frac{e^{2\la_2}}{4}\right)^n}
{\prod_{k=3}^n\frac{e^{\la_2-\la_k}}{2}}
=\frac{1}{2^{n+2}}\,\prod_{k=3}^{n} e^{\la_k},
\end{aligned}
\end{equation}
hence
\begin{equation}\label{ap28}
\begin{aligned}
\lim_{\la_2 \to\infty}e^{-(n-1)\la_2}& Z_{n-1,n}^{\textrm{inh}} 
= \frac{(-1)^{n(n-1)/2} \sinh(2\ga)  f_1(\ga)
\prod_{j=3}^n \left[e^{2\la_j}\varphi(\la_j)^n\right]  }{2^n \prod_{j=0}^{n-1}j!
\prod_{3 \le j<k\le n} \sinh(\la_j-\la_k)} \\
&\qquad \times \det \begin{pmatrix} 1 & (-2) & (-2)^2 & \dots & (-2)^{n-1} \\ 1 & (-4) & (-4)^2 & \dots & (-4)^{n-1} \\ \phi(\la_3) & \phi'(\la_3) & \phi''(\la_3) & \dots & \phi^{(n-1)}(\la_3) \\
\vdots & \vdots &\vdots & \ddots & \vdots \\  \phi(\la_n) & \phi'(\la_n) & \phi''(\la_n) & \dots & \phi^{(n-1)}(\la_n) \end{pmatrix}.
\end{aligned}
\end{equation}
Thus,
\begin{equation}\label{ap24a}
\begin{aligned}
Z_{n-2,n}^{\textrm{inh}}&=\frac{2^{n-1}}{c f_1(\ga)}\lim_{\la_{2}\to \infty} e^{-(n-1)\la_{1}}Z_{n-1,n}^{\textrm{inh}} 
= \frac{(-1)^{n(n-1)/2}\prod_{j=3}^{n} \left[e^{2\la_j}\varphi(\la_j)^n\right]}{2\prod_{j=0}^{n-1} j! \prod_{3\le j <k \le n} \sinh(\la_j-\la_k)}\\
&\times  \det \begin{pmatrix} 1 & (-2) & (-2)^2 & \dots & (-2)^{n-1} \\ 1 & (-4) & (-4)^2 & \dots & (-4)^{n-1} \\ \phi(\la_3) & \phi'(\la_3) & \phi''(\la_3) & \dots & \phi^{(n-1)}(\la_3) \\
\vdots & \vdots &\vdots & \ddots & \vdots \\  \phi(\la_n) & \phi'(\la_n) & \phi''(\la_n) & \dots & \phi^{(n-1)}(\la_n) \end{pmatrix}.
\end{aligned}
\end{equation}

Continuing in this manner $m$ times, we arrive at the formula
\begin{equation}\label{ap25a}
\begin{aligned}
Z_{n-m,n}^{\textrm{inh}} = &\frac{(-1)^{n(n-1)/2}\prod_{j=m+1}^n \big[e^{m\la_j}\varphi(\la_j)^n \big] }
{2^{1+2+\cdots+(m-1)} \prod_{j=0}^{n-1}j! \prod_{m+1 \le j<k\le n} \sinh(\la_j-\la_k)} \\
&\times\det \begin{pmatrix} 1 & (-2) & (-2)^2 & \dots & (-2)^{n-1} \\ \vdots & \vdots &\vdots & \ddots & \vdots \\ 1 & (-2m) & (-2m)^2 & \dots & (-2m)^{n-1} \\ \phi(\la_{m+1}) & \phi'(\la_{m+1}) & \phi''(\la_{m+1}) & \dots & \phi^{(n-1)}(\la_{m+1}) \\
\vdots & \vdots &\vdots & \ddots & \vdots \\  \phi(\la_n) & \phi'(\la_n) & \phi''(\la_n) & \dots & \phi^{(n-1)}(\la_n) \end{pmatrix},
\end{aligned}
\end{equation}
which proves Proposition \ref{inhomo_pf}.
\end{proof}

\end{appendix}

\bibliographystyle{plain}
\bibliography{bibliography}

\def\cydot{\leavevmode\raise.4ex\hbox{.}}
\begin{thebibliography}{10}

\bibitem{BKMM}
J.~Baik, T.~Kriecherbauer, K.~T.-R. McLaughlin, and P.~D. Miller.
\newblock {\em Discrete orthogonal polynomials}, volume 164 of {\em Annals of
  Mathematics Studies}.
\newblock Princeton University Press, Princeton, NJ, 2007.
\newblock Asymptotics and applications.

\bibitem{Bleher-Liechty09}
Pavel Bleher and Karl Liechty.
\newblock Exact solution of the six-vertex model with domain wall boundary
  conditions. {F}erroelectric phase.
\newblock {\em Comm. Math. Phys.}, 286(2):777--801, 2009.

\bibitem{Bleher-Liechty14}
Pavel Bleher and Karl Liechty.
\newblock {\em Random matrices and the six-vertex model}, volume~32 of {\em CRM
  Monograph Series}.
\newblock American Mathematical Society, Providence, RI, 2014.

\bibitem{Borodin-Boyarchenko03}
Alexei Borodin and Dmitriy Boyarchenko.
\newblock Distribution of the first particle in discrete orthogonal polynomial
  ensembles.
\newblock {\em Comm. Math. Phys.}, 234(2):287--338, 2003.

\bibitem{Foda-Wheeler12}
O.~Foda and M.~Wheeler.
\newblock Partial domain wall partition functions.
\newblock {\em J. High Energy Phys.}, (7):186, front matter+35, 2012.

\bibitem{Izergin87}
A.~G. Izergin.
\newblock Partition function of a six-vertex model in a finite volume.
\newblock {\em Dokl. Akad. Nauk SSSR}, 297(2):331--333, 1987.

\bibitem{Izergin-Coker-Korepin92}
A.~G. Izergin, D.~A. Coker, and V.~E. Korepin.
\newblock Determinant formula for the six-vertex model.
\newblock {\em J. Phys. A}, 25(16):4315--4334, 1992.

\bibitem{Johansson00}
Kurt Johansson.
\newblock Shape fluctuations and random matrices.
\newblock {\em Comm. Math. Phys.}, 209(2):437--476, 2000.

\bibitem{Koekoek-Lesky-Swarttouw10}
Roelof Koekoek, Peter~A. Lesky, and Ren{\'e}~F. Swarttouw.
\newblock {\em Hypergeometric orthogonal polynomials and their
  {$q$}-analogues}.
\newblock Springer Monographs in Mathematics. Springer-Verlag, Berlin, 2010.
\newblock With a foreword by Tom H. Koornwinder.

\bibitem{Korepin82}
V.~E. Korepin.
\newblock Calculation of norms of {B}ethe wave functions.
\newblock {\em Comm. Math. Phys.}, 86(3):391--418, 1982.

\bibitem{Kostov-Matsuo12}
Ivan Kostov and Yutaka Matsuo.
\newblock Inner products of {B}ethe states as partial domain wall partition
  functions.
\newblock {\em J. High Energy Phys.}, (10):168, front matter + 14, 2012.

\bibitem{Wang-Wong11}
X.-S. Wang and R.~Wong.
\newblock Global asymptotics of the {M}eixner polynomials.
\newblock {\em Asymptot. Anal.}, 75(3-4):211--231, 2011.

\end{thebibliography}

\end{document}